%% file: DuboisMasuzawaTixeuil.tex
\documentclass[11pt]{article}

\usepackage{graphicx,geometry,verbatim,color}
\usepackage{amsmath}
\usepackage{amssymb}
\usepackage{float}

\floatstyle{ruled}
\newfloat{algorithm}{thp}{lop}[section]
\newenvironment{proof}{\noindent\textbf{Proof}}{\hfill\qed}
\newcommand{\qed}{\hfill$\Box$}

\newtheorem{theorem}{Theorem}
\newtheorem{definition}{Definition}

\begin{document}

\title{Self-Stabilization, Byzantine Containment, and Maximizable Metrics: Necessary Conditions}

\author{Swan Dubois\protect\footnote{UPMC Sorbonne Universit\'es \& INRIA, France, swan.dubois@lip6.fr} \and Toshimitsu Masuzawa\protect\footnote{Osaka University, Japan, masuzawa@ist.osaka-u.ac.jp} \and S\'{e}bastien Tixeuil\protect\footnote{UPMC Sorbonne Universit\'es \& Institut Universitaire de France, France, sebastien.tixeuil@lip6.fr}}
\date{}

\maketitle

\begin{abstract}
Self-stabilization is a versatile approach to fault-tolerance since it permits a distributed system to recover from any transient fault that arbitrarily corrupts the contents of all memories in the system. Byzantine tolerance is an attractive feature of distributed systems that permits to cope with arbitrary malicious behaviors. 

We consider the well known problem of constructing a maximum metric tree in this context. Combining these two properties is known to induce many impossibility results. In this paper, we provide two necessary conditions to construct maximum metric tree in presence of transients and (permanent) Byzantine faults.
\end{abstract}

\section{Introduction}

The advent of ubiquitous large-scale distributed systems advocates that tolerance to various kinds of faults and hazards must be included from the very early design of such systems. \emph{Self-stabilization}~\cite{D74j,D00b,T09bc} is a versatile technique that permits forward recovery from any kind of \emph{transient} faults, while \emph{Byzantine Fault-tolerance}~\cite{LSP82j} is traditionally used to mask the effect of a limited number of \emph{malicious} faults. Making distributed systems tolerant to both transient and malicious faults is appealing yet proved difficult~\cite{DW04j,DD05c,NA02c} as impossibility results are expected in many cases.

\paragraph{Related Works}
A promizing path towards multitolerance to both transient and Byzantine faults is \emph{Byzantine containment}. For \emph{local} tasks (\emph{i.e.} tasks whose correctness can be checked locally, such as vertex coloring, link coloring, or dining philosophers), the notion of \emph{strict stabilization} was proposed~\cite{NA02c,MT07j}. Strict stabilization guarantees that there exists a \emph{containment radius} outside which the effect of permanent faults is masked, provided that the problem specification makes it possible to break the causality chain that is caused by the faults. As many problems are not local, it turns out that it is impossible to provide strict stabilization for those. To circumvent impossibility results, the weaker notion of \emph{strong stabilization} was proposed~\cite{MT06cb,DMT11j}: here, correct nodes outside the containment radius may be perturbated by the actions of Byzantine node, but only a finite number of times.

Recently, the idea of generalizing strict and strong stabilization to an area that depends on the graph topology and the problem to be solved rather than an arbitrary fixed containment radius was proposed~\cite{DMT10ca,DMT10cd} and denoted by \emph{topology aware} strict (and strong) stabilization. When maximizable metric trees are considered, \cite{DMT10ca} proposed an optimal (with respect to impossibility results) protocol for topology-aware strict stabilization, and for the simpler case of breath-first-search metric trees, \cite{DMT10cd} presented a protocol that is optimal both with respect to strict and strong variants of topology-aware stabilization. The case of optimality for topology-aware strong stabilization in the general maximal metric case remains open.

\paragraph{Our Contribution} 

In this paper, we investigate the possibility of topology-aware strong stabilization for tasks that are global (\emph{i.e.} for with there exists a causality chain of size $r$, where $r$ depends on $n$ the size of the network), and focus on the maximum metric tree problem. In more details, we provide two necessary conditions to perform Byzantine containment for maximum metric tree construction. First, we characterize a specific class of maximizable metrics (which includes breath-first-search and shortest path metrics) that prevents the exitence of strong stabilizing solutions. Then, we generalize an impossibilty result of \cite{DMT10cd} that provides a lower bound on the containmemt area for topology-aware strong stabilization.  

\section{Model and Definitions}

\subsection{State Model}

A \emph{distributed system} $S=(P,L)$ consists of a set $P=\{v_1,v_2,\ldots,v_n\}$ of processes and a set $L$ of bidirectional communication links (simply called links). A link is an unordered pair of distinct processes. A distributed system $S$ can be regarded as a graph whose vertex set is $P$ and whose link set is $L$, so we use graph terminology to describe a distributed system $S$. We use the following notations: $n=|P|$, $m=|L|$ and $d(u,v)$ denotes the shortest path between two processes $u$ and $v$ (\emph{i.e} the length of the shortest path between $u$ and $v$).

Processes $u$ and $v$ are called \emph{neighbors} if $(u,v)\in L$. The set of neighbors of a process $v$ is denoted by $N_v$. We do not assume existence of a unique identifier for each process. Instead we assume each process can distinguish its neighbors from each other by locally labelling them.

In this paper, we consider distributed systems of arbitrary topology. We assume that a single process is distinguished as a \emph{root}, and all the other processes are identical. We adopt the \emph{shared state model} as a communication model in this paper, where each process can directly read the states of its neighbors.

The variables that are maintained by processes denote process states. A process may take actions during the execution of the system. An action is simply a function that is executed in an atomic manner by the process. The action executed by each process is described by a finite set of guarded actions of the form $\langle$guard$\rangle\longrightarrow\langle$statement$\rangle$. Each guard of process $u$ is a boolean expression involving the variables of $u$ and its neighbors.

A global state of a distributed system is called a \emph{configuration} and is specified by a product of states of all processes. We define $C$ to be the set of all possible configurations of a distributed system $S$. For a process set $R \subseteq P$ and two configurations $\rho$ and $\rho'$, we denote $\rho \stackrel{R}{\mapsto} \rho'$ when $\rho$ changes to $\rho'$ by executing an action of each process in $R$ simultaneously. Notice that $\rho$ and $\rho'$ can be different only in the states of processes in $R$. For completeness of execution semantics, we should clarify the configuration resulting from simultaneous actions of neighboring processes. The action of a process depends only on its state at $\rho$ and the states of its neighbors at $\rho$, and the result of the action reflects on the state of the process at $\rho '$.

We say that a process is \emph{enabled} in a configuration $\rho$ if the guard of at least one of its actions is evaluated as true in $\rho$.

A \emph{schedule} of a distributed system is an infinite sequence of process sets. Let $Q=R^1, R^2, \ldots$  be a schedule, where $R^i \subseteq P$ holds for each $i\ (i \ge 1)$. An infinite sequence of configurations $e=\rho_0,\rho_1,\ldots$ is called an \emph{execution} from an initial configuration $\rho_0$ by a schedule $Q$, if $e$ satisfies $\rho_{i-1} \stackrel{R^i}{\mapsto} \rho_i$ for each $i\ (i \ge 1)$. Process actions are executed atomically, and we distinguish some properties on the scheduler (or daemon). A \emph{distributed daemon} schedules the actions of processes such that any subset of processes can simultaneously execute their actions. We say that the daemon is \emph{central} if it schedules action of only one process at any step. The set of all possible executions from $\rho_0\in C$ is denoted by $E_{\rho_0}$. The set of all possible executions is denoted by $E$, that is, $E=\bigcup_{\rho\in C}E_{\rho}$. We consider \emph{asynchronous} distributed systems where we can make no assumption on schedules.

In this paper, we consider (permanent) \emph{Byzantine faults}: a Byzantine process (\emph{i.e.} a Byzantine-faulty process) can make arbitrary behavior independently from its actions. If $v$ is a Byzantine process, $v$ can repeatedly change its variables arbitrarily. For a given execution, the number of faulty processes is arbitrary but we assume that the root process is never faulty.

\subsection{Self-Stabilizing Protocols Resilient to Byzantine Faults}

Problems considered in this paper are so-called \emph{static problems}, \emph{i.e.} they require the system to find static solutions. For example, the spanning-tree construction problem is a static problem, while the mutual exclusion problem is not. Some static problems can be defined by a \emph{specification predicate} (shortly, specification), $spec(v)$, for each process $v$: a configuration is a desired one (with a solution) if every process satisfies $spec(v)$. A specification $spec(v)$ is a boolean expression on variables of $P_v~(\subseteq P)$ where $P_v$ is the set of processes whose variables appear in $spec(v)$. The variables appearing in the specification are called \emph{output variables} (shortly, \emph{O-variables}). In what follows, we consider a static problem defined by specification $spec(v)$.

A \emph{self-stabilizing protocol} (\cite{D74j}) is a protocol that eventually reaches a \emph{legitimate configuration}, where $spec(v)$ holds at every process $v$, regardless of the initial configuration. Once it reaches a legitimate configuration, every process never changes its O-variables and always satisfies $spec(v)$. From this definition, a self-stabilizing protocol is expected to tolerate any number and any type of transient faults since it can eventually recover from any configuration affected by the transient faults. However, the recovery from any configuration is guaranteed only when every process correctly executes its action from the configuration, \emph{i.e.}, we do not consider existence of permanently faulty processes.

When (permanent) Byzantine processes exist, Byzantine processes may not satisfy $spec(v)$. In addition, correct processes near the Byzantine processes can be influenced and may be unable to satisfy $spec(v)$. Nesterenko and Arora~\cite{NA02c} define a \emph{strictly stabilizing protocol} as a self-stabilizing protocol resilient to unbounded number of Byzantine processes.

Given an integer $c$, a \emph{$c$-correct process} is a process defined as follows.

\begin{definition}[$c$-correct process]
A process is $c$-correct if it is correct (\emph{i.e.} not Byzantine) and located at distance more than $c$ from any Byzantine process.
\end{definition}

\begin{definition}[$(c,f)$-containment]
\label{def:cfcontained}
A configuration $\rho$ is \emph{$(c,f)$-contained} for specification $spec$ if, given at most $f$ Byzantine processes, in any execution starting from $\rho$, every $c$-correct process $v$ always satisfies $spec(v)$ and never changes its O-variables.
\end{definition}

The parameter $c$ of Definition~\ref{def:cfcontained} refers to the \emph{containment radius} defined in \cite{NA02c}. The parameter $f$ refers explicitly to the number of Byzantine processes, while \cite{NA02c} dealt with unbounded number of Byzantine faults (that is $f\in\{0\ldots n\}$).

\begin{definition}[$(c,f)$-strict stabilization]
\label{def:cfstabilizing}
A protocol is \emph{$(c,f)$-strictly stabilizing} for specification $spec$ if, given at most $f$ Byzantine processes, any execution $e=\rho_0,\rho_1,\ldots$ contains a configuration $\rho_i$ that is $(c,f)$-contained for $spec$.
\end{definition}

An important limitation of the model of \cite{NA02c} is the notion of $r$-\emph{restrictive} specifications. Intuitively, a specification is $r$-restrictive if it prevents combinations of states that belong to two processes $u$ and $v$ that are at least $r$ hops away. An important consequence related to Byzantine tolerance is that the containment radius of protocols solving those specifications is at least $r$. For some (global) problems $r$ can not be bounded by a constant. In consequence, we can show that there exists no $(c,1)$-strictly stabilizing
protocol for such a problem for any (finite) integer $c$.

\paragraph{Strong stabilization} To circumvent such impossibility results, \cite{DMT11j} defines a weaker notion than the strict stabilization. Here, the requirement to the containment radius is relaxed, \emph{i.e.} there may exist processes outside the containment radius that invalidate the specification predicate, due to Byzantine actions. However, the impact of Byzantine triggered action is limited in times: the set of Byzantine processes may only impact processes outside the containment radius a bounded number of times, even if Byzantine processes execute an infinite number of actions.

In the following of this section, we recall the formal definition of strong stabilization adopted in \cite{DMT11j}. From the states of $c$-correct processes, \emph{$c$-legitimate configurations} and \emph{$c$-stable configurations} are defined as follows.

\begin{definition}[$c$-legitimate configuration]
A configuration $\rho$ is $c$-legitimate for \emph{spec} if every $c$-correct process $v$ satisfies $spec(v)$.
\end{definition}

\begin{definition}[$c$-stable configuration]
A configuration $\rho$ is $c$-stable if every $c$-correct process never changes the values of its O-variables as long as Byzantine processes make no action.
\end{definition}

Roughly speaking, the aim of self-stabilization is to guarantee that a distributed system eventually reaches a $c$-legitimate and $c$-stable configuration. However, a self-stabilizing system can be disturbed by Byzantine processes after reaching a $c$-legitimate and $c$-stable configuration. The \emph{$c$-disruption} represents the period where $c$-correct processes are disturbed by Byzantine processes and is defined as follows 

\begin{definition}[$c$-disruption]
A portion of execution $e=\rho_0,\rho_1,\ldots,\rho_t$ ($t>1$) is a $c$-disruption if and only if the following holds:
\begin{enumerate}
\item $e$ is finite,
\item $e$ contains at least one action of a $c$-correct process for changing the value of an O-variable,
\item $\rho_0$ is $c$-legitimate for \emph{spec} and $c$-stable, and
\item $\rho_t$ is the first configuration after $\rho_0$ such that $\rho_t$ is $c$-legitimate for \emph{spec} and $c$-stable.
\end{enumerate}
\end{definition}

Now we can define a self-stabilizing protocol such that Byzantine processes may only impact processes outside the containment radius a bounded number of times, even if Byzantine processes execute an infinite number of actions.

\begin{definition}[$(t,k,c,f)$-time contained configuration]
A configuration $\rho_0$ is $(t,k,c,f)$-time contained for \emph{spec} if given at most $f$ Byzantine processes, the following properties are satisfied:
\begin{enumerate}
\item $\rho_0$ is $c$-legitimate for \emph{spec} and $c$-stable,
\item every execution starting from $\rho_0$ contains a $c$-legitimate configuration for \emph{spec} after which the values of all the O-variables of $c$-correct processes remain unchanged (even when Byzantine processes make actions repeatedly and forever), 
\item every execution starting from $\rho_0$ contains at most $t$ $c$-disruptions, and 
\item every execution starting from $\rho_0$ contains at most $k$ actions of changing the values of O-variables for each $c$-correct process.
\end{enumerate}
\end{definition}

\begin{definition}[$(t,c,f)$-strongly stabilizing protocol]
A protocol $A$ is $(t,c,f)$-strongly stabilizing if and only if starting from any arbitrary configuration, every execution involving at most $f$ Byzantine processes contains a $(t,k,c,f)$-time contained configuration that is reached after at most $l$ rounds. Parameters $l$ and $k$ are respectively the $(t,c,f)$-stabilization time and the $(t,c,f)$-process-disruption times of $A$.
\end{definition}

Note that a $(t,k,c,f)$-time contained configuration is a $(c,f)$-contained configuration when $t=k=0$, and thus, $(t,k,c,f)$-time contained configuration is a generalization (relaxation) of a $(c,f)$-contained configuration. Thus, a strongly stabilizing protocol is weaker than a strictly stabilizing one (as processes outside the containment radius may take incorrect actions due to Byzantine influence). However, a strongly stabilizing protocol is stronger than a classical self-stabilizing one (that may never meet their specification in the presence of Byzantine processes).

The parameters $t$, $k$ and $c$ are introduced to quantify the strength of fault containment, we do not require each process to know the values of the parameters.

\paragraph{Topology-aware Byzantine resilience} We saw previously that there exist a number of impossibility results on strict stabilization due to the notion of $r$-restrictives specifications. To circumvent this impossibility result, we describe here another weaker notion than the strict stabilization: the \emph{topology-aware strict stabilization} (denoted by TA strict stabilization for short) introduced by \cite{DMT10ca}. Here, the requirement to the containment radius is relaxed, \emph{i.e.} the set of processes which may be disturbed by Byzantine ones is not reduced to the union of $c$-neighborhood of Byzantine processes (\emph{i.e.} the set of processes at distance at most $c$ from a Byzantine process) but can be defined depending on the graph topology and Byzantine processes location.

In the following, we give formal definition of this new kind of Byzantine containment. From now, $B$ denotes the set of Byzantine processes and $S_B$ (which is function of $B$) denotes a subset of $V$ (intuitively, this set gathers all processes which may be disturbed by Byzantine processes).

\begin{definition}[$S_{B}$-correct node]
A node is \emph{$S_{B}$-correct} if it is a correct node (\emph{i.e.} not Byzantine) which not belongs to $S_{B}$.
\end{definition}

\begin{definition}[$S_{B}$-legitimate configuration]
A configuration $\rho$ is \emph{$S_{B}$-legitimate} for $spec$ if every $S_{B}$-correct node $v$ is legitimate for $spec$ (\emph{i.e.} if $spec(v)$ holds).
\end{definition}

\begin{definition}[$(S_{B},f)$-topology-aware containment]
\label{def:SfTAcontained}
A configuration $\rho_{0}$ is \emph{$(S_{B},f)$-topology-aware contained} for specification $spec$ if, given at most $f$ Byzantine processes, in any execution $e=\rho_0,\rho_1,\ldots$, every configuration is $S_{B}$-legitimate and every $S_B$-correct process never changes its O-variables. 
\end{definition}

The parameter $S_{B}$ of Definition~\ref{def:SfTAcontained} refers to the \emph{containment area}. Any process which belongs to this set may be infinitely disturbed by Byzantine processes. The parameter $f$ refers explicitly to the number of Byzantine processes.

\begin{definition}[$(S_{B},f)$-topology-aware strict stabilization]
\label{def:SfTAStrictstabilizing}
A protocol is \emph{$(S_{B},f)$-topology-aware strictly stabilizing} for specification $spec$ if, given at most $f$ Byzantine processes, any execution $e=\rho_0,\rho_1,\ldots$ contains a configuration $\rho_i$ that is $(S_{B},f)$-topology-aware contained for $spec$.
\end{definition}

Note that, if $B$ denotes the set of Byzantine processes and $S_{B}=\left\{v\in V|\underset{b\in B}{min}\left(d(v,b)\right)\leq c\right\}$, then a $(S_{B},f)$-topology-aware strictly stabilizing protocol is a $(c,f)$-strictly stabilizing protocol. Then, the concept of topology-aware strict stabilization is a generalization of the strict stabilization. However, note that a TA strictly stabilizing protocol is stronger than a classical self-stabilizing protocol (that may never meet their specification in the presence of Byzantine processes). The parameter $S_{B}$ is introduced to quantify the strength of fault containment, we do not require each process to know the actual definition of the set.

Similarly to topology-aware strict stabilization, we can weaken the notion of strong stabilization using the notion of containment area. This idea was introduced by \cite{DMT10cd}. We recall in the following the formal definition of this concept.

\begin{definition}[$S_B$-stable configuration]
A configuration $\rho$ is $S_B$-stable if every $S_B$-correct process never changes the values of its O-variables as long as Byzantine processes make no action.
\end{definition}

\begin{definition}[$S_{B}$-TA-disruption]
A portion of execution $e=\rho_0,\rho_1,\ldots,\rho_t$ ($t>1$) is a $S_{B}$-TA-disruption if and only if the followings hold:
\begin{enumerate}
\item $e$ is finite,
\item $e$ contains at least one action of a $S_{B}$-correct process for changing the value of an O-variable,
\item $\rho_0$ is $S_{B}$-legitimate for $spec$ and $S_B$-stable, and
\item $\rho_t$ is the first configuration after $\rho_0$ such that $\rho_t$ is $S_{B}$-legitimate for $spec$ and $S_B$-stable.
\end{enumerate}
\end{definition}

\begin{definition}[$(t,k,S_{B},f)$-TA time contained configuration]
A configuration $\rho_0$ is $(t,k,S_{B},$ $f)$-TA time contained for \emph{spec} if given at most $f$ Byzantine processes, the following properties are satisfied:
\begin{enumerate}
\item $\rho_0$ is $S_{B}$-legitimate for \emph{spec} and $S_B$-stable,
\item every execution starting from $\rho_0$ contains a $S_B$-legitimate configuration for \emph{spec} after which the values of all the O-variables of $S_B$-correct processes remain unchanged (even when Byzantine processes make actions repeatedly and forever), 
\item every execution starting from $\rho_0$ contains at most $t$ $S_B$-TA-disruptions, and 
\item every execution starting from $\rho_0$ contains at most $k$ actions of changing the values of O-variables for each $S_B$-correct process.
\end{enumerate}
\end{definition}

\begin{definition}[$(t,S_{B},f)$-TA strongly stabilizing protocol]
A protocol $A$ is $(t,S_{B},f)$-TA\\ strongly stabilizing if and only if starting from any arbitrary configuration, every execution involving at most $f$ Byzantine processes contains a $(t,k,S_{B},f)$-TA-time contained configuration that is reached after at most $l$ actions of each $S_{B}$-correct node. Parameters $l$ and $k$ are respectively the $(t,S_{B},f)$-stabilization time and the $(t,S_{B},f)$-process-disruption time of $A$.
\end{definition}

\section{Maximum Metric Tree Construction}

\subsection{Definition and Specification}

In this work, we deal with maximum (routing) metric trees as defined in \cite{GS03j}. Informally, the goal of a routing protocol is to construct a tree that simultaneously maximizes the metric values of all of the nodes with respect to some total ordering $\prec$. In the following, we recall all definitions and notations introduced in \cite{GS03j}. 

\begin{definition}[Routing metric]
A \emph{routing metric} (or just \emph{metric}) is a five-tuple $(M,W,met,mr,$ $\prec)$ where:
\begin{enumerate}
\item $M$ is a set of metric values,
\item $W$ is a set of edge weights,
\item $met$ is a metric function whose domain is $M\times W$ and whose range is $M$,
\item $mr$ is the maximum metric value in $M$ with respect to $\prec$ and is assigned to the root of the system,
\item $\prec$ is a less-than total order relation over $M$ that satisfies the following three conditions for arbitrary metric values $m$, $m'$, and $m''$ in $M$:
\begin{enumerate}
\item irreflexivity: $m\not\prec m$,
\item transitivity : if $m\prec m'$ and $m'\prec m''$ then $m\prec m''$,
\item totality: $m\prec m'$ or $m'\prec m$ or $m=m'$.
\end{enumerate}
\end{enumerate}
Any metric value $m\in M\setminus\{mr\}$ satisfies the \emph{utility condition} (that is, there exist $w_0,\ldots,w_{k-1}$ in $W$ and $m_0=mr,m_1,\ldots,m_{k-1},m_{k}=m$ in $M$ such that $\forall i\in\{1,\ldots,k\},m_i=met(m_{i-1},w_{i-1})$).
\end{definition}

For instance, we provide the definition of four classical metrics with this model: the shortest path metric ($\mathcal{SP}$), the flow metric ($\mathcal{F}$), and the reliability metric ($\mathcal{R}$). Note also that we can modelise the construction of a spanning tree with no particular constraints in this model using the metric $\mathcal{NC}$ described below and the construction of a BFS spanning tree using the shortest path metric ($\mathcal{SP}$) with $W_1=\{1\}$ (we denoted this metric by $\mathcal{BFS}$ in the following).

\[\begin{array}{rclrcl}
\mathcal{SP}&=&(M_1,W_1,met_1,mr_1,\prec_1)&\mathcal{F}&=&(M_2,W_2,met_2,mr_2,\prec_2)\\
\text{where}& & M_1=\mathbb{N}&\text{where}& & mr_2\in\mathbb{N}\\
&& W_1=\mathbb{N}&&& M_2=\{0,\ldots,mr_2\}\\
&& met_1(m,w)=m+w&&& W_2=\{0,\ldots,mr_2\}\\
&& mr_1=0&&& met_2(m,w)=min\{m,w\}\\
&& \prec_1 \text{ is the classical }>\text{ relation}&&& \prec_2 \text{ is the classical }<\text{ relation}
\end{array}\]
\[\begin{array}{rclrcl}
\mathcal{R}&=&(M_3,W_3,met_3,mr_3,\prec_3) & \mathcal{NC}&=&(M_4,W_4,met_4,mr_4,\prec_4)\\
\text{where}& & M_3=[0,1] & \text{where}& & M_4=\{0\} \\
&& W_3=[0,1] &&& W_4=\{0\}\\
&& met_3(m,w)=m*w &&& met_4(m,w)=0\\
&& mr_3=1 &&& mr_4=0\\
&& \prec_3 \text{ is the classical }<\text{ relation} &&& \prec_4 \text{ is the classical }<\text{ relation}
\end{array}\]

\begin{definition}[Assigned metric]
An \emph{assigned metric} over a system $S$ is a six-tuple $(M,W,met,$ $mr,\prec,wf)$ where $(M,W,met,mr,\prec)$ is a metric and $wf$ is a function that assigns to each edge of $S$ a weight in $W$.
\end{definition}

Let a rooted path (from $v$) be a simple path from a process $v$ to the root $r$. The next set of definitions are with respect to an assigned metric $(M,W,met,mr,\prec,wf)$ over a given system $S$.

\begin{definition}[Metric of a rooted path]
The \emph{metric of a rooted path} in $S$ is the prefix sum of $met$ over the edge weights in the path and $mr$.
\end{definition}

For example, if a rooted path $p$ in $S$ is $v_k,\ldots,v_0$ with $v_0=r$, then the metric of $p$ is $m_k=met(m_{k-1},wf(\{v_k,v_{k-1}\})$ with $\forall i\in\{1,\ldots,k-1\},m_i=met(m_{i-1},wf(\{v_i,v_{i-1}\})$ and $m_0=mr$.

\begin{definition}[Maximum metric path]
A rooted path $p$ from $v$ in $S$ is called a \emph{maximum metric path} with respect to an assigned metric if and only if for every other rooted path $q$ from $v$ in $S$, the metric of $p$ is greater than or equal to the metric of $q$ with respect to the total order $\prec$. 
\end{definition}
 
\begin{definition}[Maximum metric of a node]
The \emph{maximum metric of a node} $v\neq r$ (or simply \emph{metric value} of $v$) in $S$ is defined by the metric of a maximum metric path from $v$. The maximum metric of $r$ is $mr$. 
\end{definition}

\begin{definition}[Maximum metric tree]
A spanning tree $T$ of $S$ is a \emph{maximum metric tree} with respect to an assigned metric over $S$ if and only if every rooted path in $T$ is a maximum metric path in $S$ with respect to the assigned metric.
\end{definition}

The goal of the work of \cite{GS03j} is the study of metrics that always allow the construction of a maximum metric tree. More formally, the definition follows.

\begin{definition}[Maximizable metric]
A metric is \emph{maximizable} if and only if for any assignment of this metric over any system $S$, there is a maximum metric tree for $S$ with respect to the assigned metric.
\end{definition}

Given a maximizable metric $\mathcal{M}=(M,W,mr,met,\prec)$, the aim of this work is to study the construction of a maximum metric tree with respect to $\mathcal{M}$ which spans the system in a self-stabilizing way in a system subject to permanent Byzantine failures. It is obvious that these Byzantine processes may disturb some correct processes. It is why, we relax the problem in the following way: we want to construct a maximum metric forest with respect to $\mathcal{M}$. The root of any tree of this forest must be either the real root or a Byzantine process. 

Each process $v$ has two O-variables: a pointer to its parent in its tree ($prnt_v\in N_v\cup\{\bot\}$) and a level which stores its current metric value ($level_v\in M$). Obviously, Byzantine process may disturb (at least) their neighbors. We use the following specification of the problem.

We introduce new notations as follows. Given an assigned metric $(M,W,met,mr,\prec,wf)$ over the system $S$ and two processes $u$ and $v$, we denote by $\mu(u,v)$ the maximum metric of node $u$ when $v$ plays the role of the root of the system. If $u$ and $v$ are neighbors, we denote by $w_{u,v}$ the weight of the edge $\{u,v\}$ (that is, the value of $wf(\{u,v\})$).

\begin{definition}[$\mathcal{M}$-path]
Given an assigned metric $\mathcal{M}=(M,W,mr,met,\prec,wf)$ over a system $S$, a path $(v_0,\ldots,v_k)$ ($k\geq 1$) of $S$ is a \emph{$\mathcal{M}$-path} if and only if:
\begin{enumerate}
\item $prnt_{v_0}=\bot$, $level_{v_0}=0$, and $v_0\in B\cup\{r\}$,
\item $\forall i\in\{1,\ldots,k\}, prnt_{v_i}=v_{i-1}$ and $level_{v_i}=met(level_{v_{i-1}},w_{v_i,v_{i-1}})$,
\item $\forall i\in\{1,\ldots,k\}, met(level_{v_{i-1}},w_{v_i,v_{i-1}})=\underset{u\in N_v}{max_\prec}\{met(level_{u},w_{v_i,u})\}$, and
\item $level_{v_{k}}=\mu(v_k,v_0)$.
\end{enumerate}
\end{definition}

We define the specification predicate $spec(v)$ of the maximum metric tree construction with respect to a maximizable metric $\mathcal{M}$ as follows.
\[spec(v) : \begin{cases}
 prnt_v = \bot \text{ and }  level_v = 0 \text{ if } v \text{ is the root } r \\
 \text{there exists a }\mathcal{M}\text{-path } (v_0,\ldots,v_k) \text{ such that } v_k=v \text{ otherwise}
\end{cases}\]

\subsection{Previous results}

In this section, we summarize known results about maximum metric tree construction. The first interesting result about maximizable metrics is due to \cite{GS03j} that provides a fully characterization of maximizable metrics as follow.

\begin{definition}[Boundedness]
A metric $(M,W,met,mr,\prec)$ is \emph{bounded} if and only if: $\forall m \in M,\forall w\in W, met(m,w)\prec m \text{ or }met(m,w)=m$
\end{definition}

\begin{definition}[Monotonicity]
A metric $(M,W,met,mr,\prec)$ is \emph{monotonic} if and only if: $\forall (m,$ $m')\in M^2,\forall w\in W, m\prec m'\Rightarrow (met(m,w)\prec met(m',w)\text{ or }met(m,w)=met(m',w))$
\end{definition}

\begin{theorem}[Characterization of maximizable metrics \cite{GS03j}]
A metric is maximizable if and only if this metric is bounded and monotonic.
\end{theorem}

Secondly, \cite{GS99c} provides a self-stabilizing protocol to construct a maximum metric tree with respect to any maximizable metric. Now, we focus on self-stabilizating solutions resilient to Byzantine faults. Following discussion of Section 2, it is obvious that there exists no strictly stabilizing protocol for this problem. If we consider the weaker notion of topology-aware strict stabilization, \cite{DMT10ca} defines the best containment area as:

\[S_{B}=\left\{v\in V\setminus B\left|\mu(v,r)\preceq max_\prec\{\mu(v,b),b\in B\}\right.\right\}\setminus\{r\}\]

Intuitively, $S_B$ gathers correct processes that are closer (or at equal distance) from a Byzantine process than the root according to the metric. Moreover, \cite{DMT10ca} proves that the algorithm introduced for the maximum metric spanning tree construction in \cite{GS99c} performed this optimal containment area. More formally, \cite{DMT10ca} proves the following results.

\begin{theorem}[\cite{DMT10ca}]\label{th:impTAstrict}
Given a maximizable metric $\mathcal{M}=(M,W,mr,met,\prec)$, even under the central daemon, there exists no $(A_B,1)$-TA-strictly stabilizing protocol for maximum metric spanning tree construction with respect to $\mathcal{M}$ where $A_B\varsubsetneq S_B$.
\end{theorem}

\begin{theorem}[\cite{DMT10ca}]\label{th:SSMAXstrict}
Given a maximizable metric $\mathcal{M}=(M,W,mr,met,\prec)$, the protocol of \cite{GS99c} is a $(S_B,n-1)$-TA strictly stabilizing protocol for maximum metric spanning tree construction with respect to $\mathcal{M}$.
\end{theorem}

Some others works try to circumvent the impossibility result of strict stabilization using the concept ot strong stabilization but do not provide results for any maximizable metric. Indeed, \cite{DMT11j} proves the following result about spanning tree.

\begin{theorem}[\cite{DMT11j}]
There exists a $(t,0,n-1)$-strongly stabilizing protocol for maximum metric spanning tree construction with respect to $\mathcal{NC}$ (that is, for a spanning tree with no particular constraints) with a finite $t$.
\end{theorem}
 
On the other hand, regarding BFS spanning tree construction, \cite{DMT10cd} proved the following impossibility result.

\begin{theorem}[\cite{DMT10cd}]
Even under the central daemon, there exists no $(t,c,1)$-strongly stabilizing protocol for maximum metric spanning tree construction with respect to $\mathcal{BFS}$ where $t$ and $c$ are two finite integers.
\end{theorem}

These two results motivate our result related to strong stabilization in the general case (see Section 4.1) that proves a necessary condition on the maximizable metric to allow strong stabilization.

Now, if we focus on topology-aware strong stabilization, \cite{DMT10cd} proved the following results.

\begin{theorem}[\cite{DMT10cd}]\label{th:impTAStrongBFS}
Even under the central daemon, there exists no $(t,A_B^*,1)$-TA strongly stabilizing protocol for maximum metric spanning tree construction with respect to $\mathcal{BFS}$ where $A_B^*\varsubsetneq \{v\in V|\underset{b\in B}{min}(d(v,b))<d(r,v)\}$ and $t$ is a finite integer.
\end{theorem}

\begin{theorem}[\cite{DMT10cd}]
The protocol of \cite{HC92j} is a $(t,S_B^*,n-1)$-TA strongly stabilizing protocol for maximum metric spanning tree construction with respect to $\mathcal{BFS}$ where $t$ is a finite integer and $S_B^*=\{v\in V|\underset{b\in B}{min}(d(v,b))<d(r,v)\}$.
\end{theorem}

In the following, we generalize the Theorem \ref{th:impTAStrongBFS} to any maximizable metric (see Section 4.2).

\section{Necessary conditions}

In this section, we provide our necessary conditions about containment radius (respectively area) of any strongly stabilizing (respectively TA strongly stabilizing) protocol for the maximum metric tree construction.

\subsection{Strong Stabilization}

We introduce here some new definitions to characterize some important properties of maximizable metrics that are used in the following.

\begin{definition}[Strictly decreasing metric]
A metric $\mathcal{M}=(M,W,mr,met,\prec)$ is \emph{strictly decreasing} if, for any metric value $m\in M$, the following property holds: either $\forall w\in W,met(m,w)\prec m$ or $\forall w\in W,met(m,w)=m$.
\end{definition}

\begin{definition}[Fixed point]
A metric value $m$ is a \emph{fixed point} of a metric $\mathcal{M}=(M,W,mr,met,\prec)$ if $m\in M$ and if for any value $w\in W$, we have: $met(m,w)=m$.
\end{definition}

Then, we define a specific class of maximizable metrics and we prove that it is possible to construct a maximum metric tree in a strongly-stabilizing way only if we consider such a metric.

\begin{definition}[Strongly maximizable metric]
A maximizable metric $\mathcal{M}=(M,W,mr,met,\prec)$ is strongly maximizable if and only if $|M|=1$ or if the following properties holds: 
\begin{itemize}
\item $|M|\geq 2$,
\item $\mathcal{M}$ is strictly decreasing, and 
\item $\mathcal{M}$ has one and only one fixed point.
\end{itemize}
\end{definition}

Note that $\mathcal{NC}$ is a strongly maximizable metric (since $|M_4|=1$) whereas $\mathcal{BFS}$ or $\mathcal{SP}$ are not (since the first one has no fixed point, the second is not strictly decreasing). If we consider the metric $\mathcal{MET}$ defined below, we can show that $\mathcal{MET}$ is a strongly maximizable metric such that $|M|\geq 2$.

\[\begin{array}{rcl}
\mathcal{MET}&=&(M_5,W_5,met_5,mr_5,\prec_5)\\
\text{where}& & M_5=\{0,1,2,3\}\\
&& W_5=\{1\}\\
&& met_5(m,w)=max\{0,m-w\}\\
&& mr_5=3\\
&& \prec_5 \text{ is the classical }<\text{ relation}
\end{array}\]

Now, we can state our first necessary condition.

\begin{theorem}\label{th:necessarConditionStrong}
Given a maximizable metric $\mathcal{M}=(M,W,mr,met,\prec)$, even under the central daemon, there exists no $(t,c,1)$-strongly stabilizing protocol for maximum metric spanning tree construction with respect to $\mathcal{M}$ for any finite integer $t$ if:
\[\left\{\begin{array}{l}
\mathcal{M} \mbox{ is not a strongly maximizable metric}\\
\mbox{ or }\\
c<|M|-2
\end{array}\right.\]
\end{theorem}

\begin{proof}
We prove this result by contradiction. We assume so that $\mathcal{M}=(M,W,mr,met,\prec)$ is a maximizable metric such that there exist a finite integer $t$ and a protocol $\mathcal{P}$ that is a $(t,c,1)$-strongly stabilizing protocol for maximum metric spanning tree construction with respect to $\mathcal{M}$. We distinguish the following cases (note that they are exhaustive):
\begin{description}
\item[Case 1:] $\mathcal{M}$ is a strongly maximizing metric and $c<|M|-2$.

As $c\geq 0$, we know that $|M|\geq 2$ and by definition of a strongly stabilizing metric, $\mathcal{M}$ is strictly decreasing, and $\mathcal{M}$ has one and only one fixed point.

\begin{figure}[t]
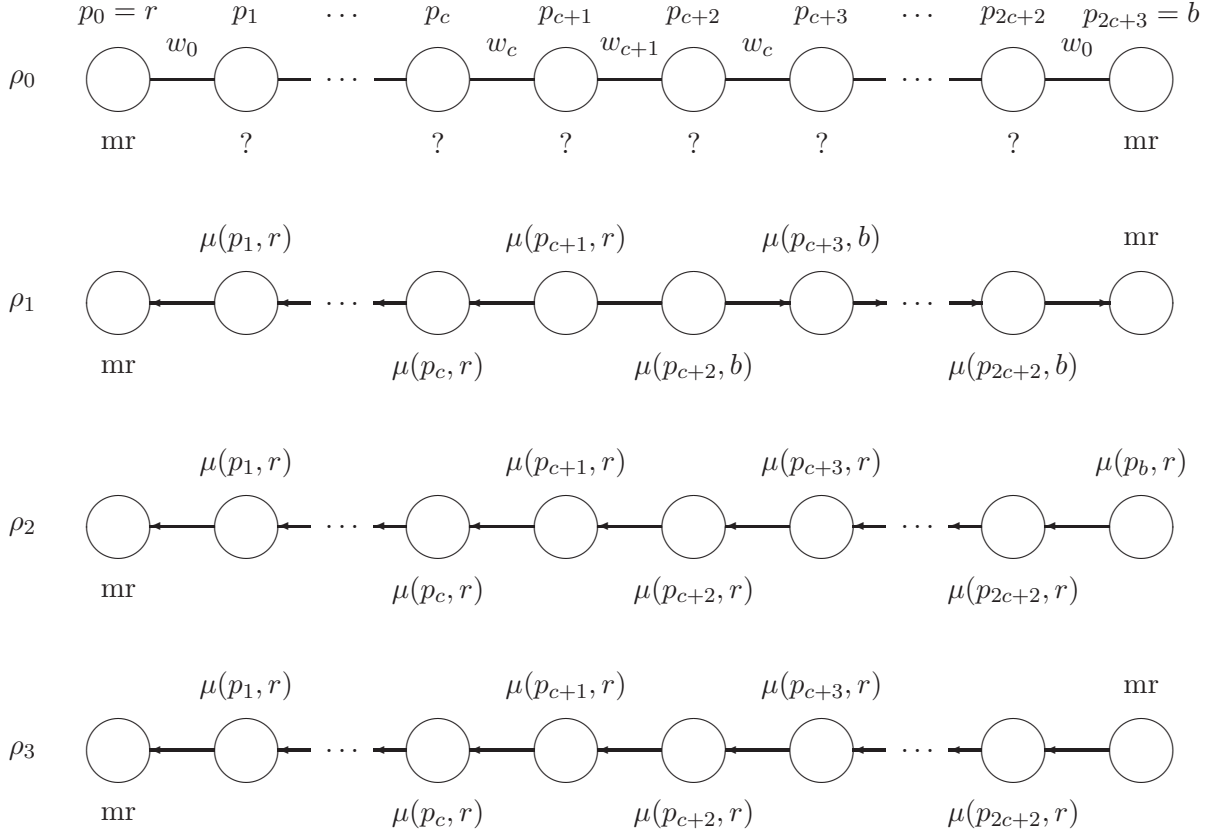

\noindent \begin{centering} \include{possStrongCase1}
  \par\end{centering}
 \caption{Configurations used in proof of Theorem \ref{th:necessarConditionStrong}, case 1.}
\label{fig:possStrongCase1}
\end{figure}

By assumption on $\mathcal{M}$, we know that there exist $c+3$ distinct metric values $m_0=mr,m_1,\ldots,$ $m_{c+2}$ in $M$ and $w_0,w_1,\ldots,w_{c+1}$ in $W$ such that: $\forall i\in\{1,\ldots,c+2\},m_i=met(m_{i-1},w_{i-1})\prec m_{i-1}$.

Let $S=(V,E,\mathcal{W})$ be the following weighted system $V=\{p_0=r,p_1,\ldots,p_{2c+2},p_{2c+3}=b\}$, $E=\{\{p_i,p_{i+1}\},i\in\{0,\ldots,2c+2\}\}$ and $\forall i\in\{0,c+1\},w_{p_i,p_{i+1}}=w_{p_{2c+3-i},p_{2c+2-i}}=w_i$. Note that the choice $w_{p_{c+1},p_{c+2}}=w_{c+1}$ ensures us the following property when $level_r=level_b=mr$: $\mu(p_{c+1},b)\prec\mu(p_{c+1},r)$ (and by symmetry, $\mu(p_{c+2},r)\prec\mu(p_{c+2},b)$). Process $p_0$ is the real root and process $b$ is a Byzantine one. Note that the construction of $\mathcal{W}$ ensures the following properties when $level_r=level_b=mr$: $\forall i\in\{1,\ldots,c+1\},\mu(p_i,r)=\mu(p_{2c+3-i},b)$, $\mu(p_i,b)\prec\mu(p_i,r)$ and $\mu(p_{2c+3-i},r)\prec\mu(p_{2c+3-i},b)$.

Assume that the initial configuration $\rho_0$ of $S$ satisfies: $prnt_r=prnt_b=\bot$, $level_r=level_b=mr$, and other variables of $b$ (if any) are identical to those of $r$ (see Figure \ref{fig:possStrongCase1}, variables of other processes may be arbitrary). Assume now that $b$ takes exactly the same actions as $r$ (if any) immediately after $r$. Then, by symmetry of the execution and by convergence of $\mathcal{P}$ to $spec$, we can deduce that the system reaches in a finite time a configuration $\rho_1$ (see Figure \ref{fig:possStrongCase1}) in which: $\forall i\in\{1,\ldots,c+1\}, prnt_{p_i}=p_{i-1}$, $level_{p_i}=\mu(p_i,r)=m_i$ and $\forall i\in\{c+2,\ldots,2c+2\},prnt_{p_i}=p_{i+1}$ and $level_{p_i}=\mu(p_{i},b)=m_{2c+3-i}$ (because this configuration is the only one in which all correct process $v$ satisfies $spec(v)$ when $prnt_r=prnt_b=\bot$ and $level_r=level_b=mr$ by construction of $\mathcal{W}$). Note that $\rho_1$ is $c$-legitimate and $c$-stable.

Assume now that the Byzantine process acts as a correct process and executes correctly its algorithm. Then, by convergence of $\mathcal{P}$ in fault-free systems (remember that a strongly-stabilizing algorithm is a special case of self-stabilizing algorithm), we can deduce that the system reach in a finite time a configuration $\rho_2$ (see Figure \ref{fig:possStrongCase1}) in which: $\forall i\in\{1,\ldots,2c+3\},prnt_{p_i}=p_{i-1}$ and $level_{p_i}=\mu(p_i,r)$ (because this configuration is the only one in which all process $v$ satisfies $spec(v)$). Note that the portion of execution between $\rho_1$ and $\rho_2$ contains at least one $c$-perturbation ($p_{c+2}$ is a $c$-correct process and modifies at least once its O-variables) and that $\rho_2$ is $c$-legitimate and $c$-stable.

Assume now that the Byzantine process $b$ takes the following state: $prnt_{b}=\bot$ and $level_b=mr$. This step brings the system into configuration $\rho_3$ (see Figure \ref{fig:possStrongCase1}). From this configuration, we can repeat the execution we constructed from $\rho_0$. By the same token, we obtain an execution of $\mathcal{P}$ which contains $c$-legitimate and $c$-stable configurations (see $\rho_1$) and an infinite number of $c$-perturbation which contradicts the $(t,c,1)$-strong stabilization of $\mathcal{P}$.

\item[Case 2:] $\mathcal{M}$ is not strictly decreasing.

\begin{figure}[t]
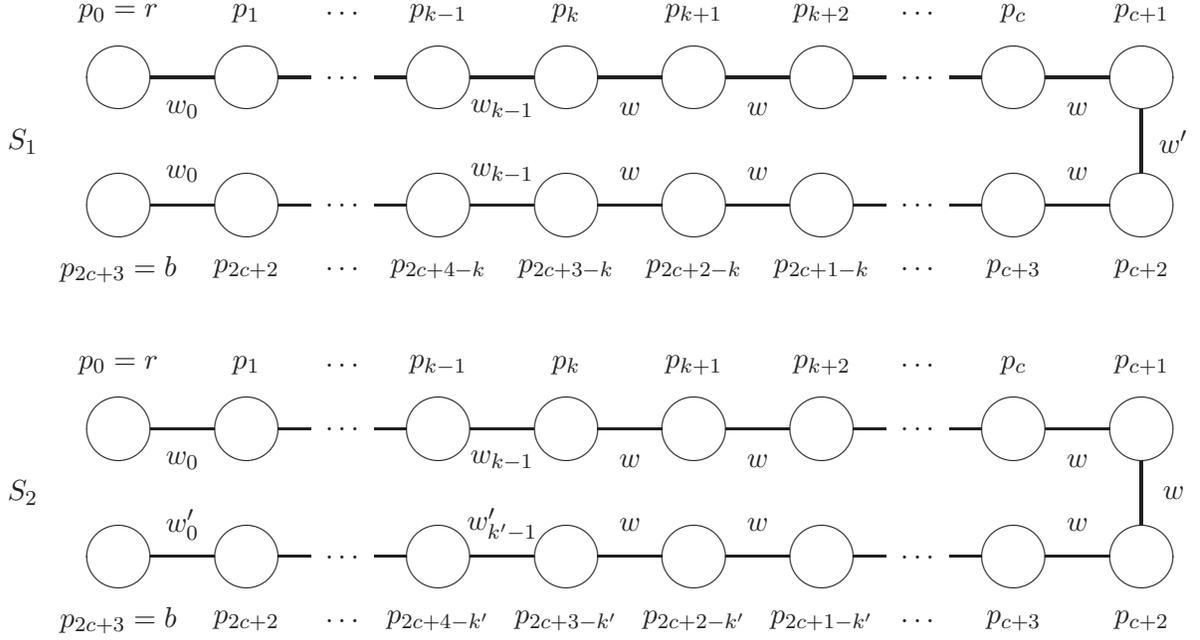

\noindent \begin{centering} \include{possStrongCase23}
  \par\end{centering}
 \caption{Configurations used in proof of Theorem \ref{th:necessarConditionStrong}, cases 2 and 3.}
\label{fig:possStrongCase23}
\end{figure}

By definition, we know that $\mathcal{M}$ is not a strongly maximizable metric. Hence, we have $|M|\geq 2$. Then, the definition of a strictly decreasing metric implies that there exists a metric value $m\in M$ such that: $\exists w\in W,$ $met(m,w)=m$ and $\exists w'\in W,m'=met(m,w')\prec m$ (and thus $m$ is not a fixed point of $\mathcal{M}$). By the utility condition on $M$, we know that there exists a sequence of metric values $m_0=mr,m_1,\ldots,m_l=m$ in $M$ and $w_0,w_1,\ldots,w_{l-1}$ in $W$ such that $\forall i\in\{1,\ldots,l\},m_i=met(m_{i-1},w_{i-1})$. Denote by $k$ the length of the shortest such sequence. Note that this implies that $\forall i\in\{1,\ldots,k\},m_i\prec m_{i-1}$ (otherwise we can remove $m_i$ from the sequence and this is contradictory with the construction of $k$). We distinguish the following cases:
\begin{description}
\item[Case 2.1:] $k\geq c+2$.\\
We can use the same token as case 1 above by using $w'$ instead of $w_{c+1}$ in the case where $k=c+2$ (since we know that $met(m,w')\prec m$).
\item[Case 2.2:] $k< c+2$.\\
Let $S_1=(V,E,\mathcal{W})$ be the following weighted system $V=\{p_0=r,p_1,\ldots,p_{2c+2},p_{2c+3}=b\}$, $E=\{\{p_i,p_{i+1}\},i\in\{0,\ldots,2c+2\}\}$, $\forall i\in\{0,\ldots,k-1\},w_{p_i,p_{i+1}}=w_{p_{2c+3-i},p_{2c+2-i}}=w_i$, $\forall i\in\{k,\ldots,c\},w_{p_i,p_{i+1}}=w_{p_{2c+3-i},p_{2c+2-i}}=w$ and $w_{p_{c+1},p_{c+2}}=w'$ (see Figure \ref{fig:possStrongCase23}). Note that this choice ensures us the following property when $level_r=level_b=mr$: $\mu(p_{c+1},b)\prec\mu(p_{c+1},r)$ (and by symmetry, $\mu(p_{c+2},r)\prec\mu(p_{c+2},b)$). Process $p_0$ is the real root and process $b$ is a Byzantine one. Note that the construction of $\mathcal{W}$ ensures the following properties when $level_r=level_b=mr$: $\forall i\in\{1,\ldots,c+1\},\mu(p_i,r)=\mu(p_{2c+3-i},b)$, $\mu(p_i,b)\prec\mu(p_i,r)$ and $\mu(p_{2c+3-i},r)\prec\mu(p_{2c+3-i},b)$.

This construction allows us to follow the same proof as in case 1 above.
\end{description}

\item[Case 3:] $\mathcal{M}$ has no or more than two fixed point, and is strictly decreasing.

If $\mathcal{M}$ has no fixed point and is strictly decreasing, then $|M|$ is not finite and then, we can apply the result of case 1 above since $c$ is a finite integer.

If $\mathcal{M}$ has two or more fixed points and is strictly decreasing, denote by $\Upsilon$ and $\Upsilon'$ two fixed points of $\mathcal{M}$. Without loss of generality, assume that $\Upsilon\prec\Upsilon'$. By the utility condition on $M$, we know that there exists sequences of metric values $m_0=mr,m_1,\ldots,m_l=\Upsilon$ and $m'_0=mr,m'_1,\ldots,m'_{l'}=\Upsilon'$ in $M$ and $w_0,w_1,\ldots,w_{l-1}$ and $w'_0,w'_1,\ldots,w'_{l'-1}$ in $W$ such that $\forall i\in\{1,\ldots,l\},m_i=met(m_{i-1},w_{i-1})$ and $\forall i\in\{1,\ldots,l'\},m'_i=met(m'_{i-1},w'_{i-1})$. Denote by $k$ and $k'$ the length of shortest such sequences. Note that this implies that $\forall i\in\{1,\ldots,k\},m_i\prec m_{i-1}$ and $\forall i\in\{1,\ldots,k'\},m'_i\prec m'_{i-1}$ (otherwise we can remove $m_i$ or $m'_i$ from the corresponding sequence). We distinguish the following cases:
\begin{description}
\item[Case 3.1:] $k>c+2$ or $k'>c+2$.\\
Without loss of generality, assume that $k>c+2$ (the second case is similar). We can use the same token as case 1 above.
\item[Case 3.2:] $k\leq c+2$ and $k'\leq c+2$.\\
Let $w$ be an arbitrary value of $W$. Let $S_2=(V,E,\mathcal{W})$ be the following weighted system $V=\{p_0=r,p_1,\ldots,p_{2c+2},p_{2c+3}=b\}$, $E=\{\{p_i,p_{i+1}\},i\in\{0,\ldots,2c+2\}\}$, $\forall i\in\{0,k-1\},w_{p_i,p_{i+1}}=w_i$, $\forall i\in\{0,k'-1\},w_{p_{2c+3-i},p_{2c+2-i}}=w'_i$ and $\forall i\in\{k,2c+2-k'\}, w_{p_{i},p_{i+1}}=w$ (see Figure \ref{fig:possStrongCase23}). Note that this choice ensures us the following property when $level_r=level_b=mr$: $\mu(p_{c+1},r)=\Upsilon\prec\Upsilon' =\mu(p_{c+1},b)$ and $\mu(p_{c+2},r)=\Upsilon\prec\Upsilon'=\mu(p_{c+2},b)$. Process $p_0$ is the real root and process $b$ is a Byzantine one. 

This construction allows us to follow a similar proof as in case 1 above (note that any process $u$ which satisfies $\mu(u,r)\prec\Upsilon'$ will be disturb infinitely often, in particular at least $p_{c+1}$ and $p_{c+2}$ which contradicts the $(t,c,1)$-strong stabilization of $\mathcal{P}$).
\end{description}
\end{description}
In any case, we show that there exists a system which contradicts the $(t,c,1)$-strong stabilization of $\mathcal{P}$ that ends the proof. 
\end{proof}

\subsection{Topology Aware Strong Stabilization}

First, we generalize the set $S_B^*$ previously defined for the $\mathcal{BFS}$ metric in \cite{DMT10cd} to any maximizable metric $\mathcal{M}=(M,W,mr,met,\prec)$.

\[S_{B}^*=\left\{v\in V\setminus B\left|\mu(v,r)\prec\underset{b\in B}{max_\prec}\{\mu(v,b)\}\right.\right\}\]

Intuitively, $S_B^*$ gathers the set of corrects processes that are strictly closer (according to $\mathcal{M}$) to a Byzantine process than the root. Figures from \ref{fig:ExSP} to \ref{fig:ExReliability} provide some examples of containment areas with respect to several maximizable metrics and compare it to $S_B$, the optimal containment area for TA strict stabilization. 

\begin{figure}[t]
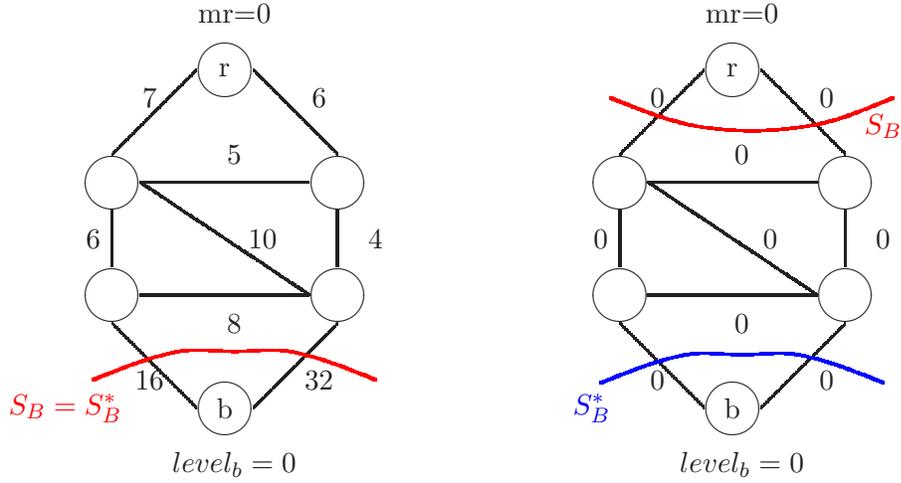

\noindent \begin{centering} \include{ExSP}
  \par\end{centering}
 \caption{Examples of containment areas for SP spanning tree construction.}
\label{fig:ExSP}
\end{figure}

\begin{figure}[t]
\noindent \begin{centering} \include{ExFlow}
  \par\end{centering}
 \caption{Examples of containment areas for flow spanning tree construction.}
\label{fig:ExFlow}
\end{figure}

\begin{figure}[t]
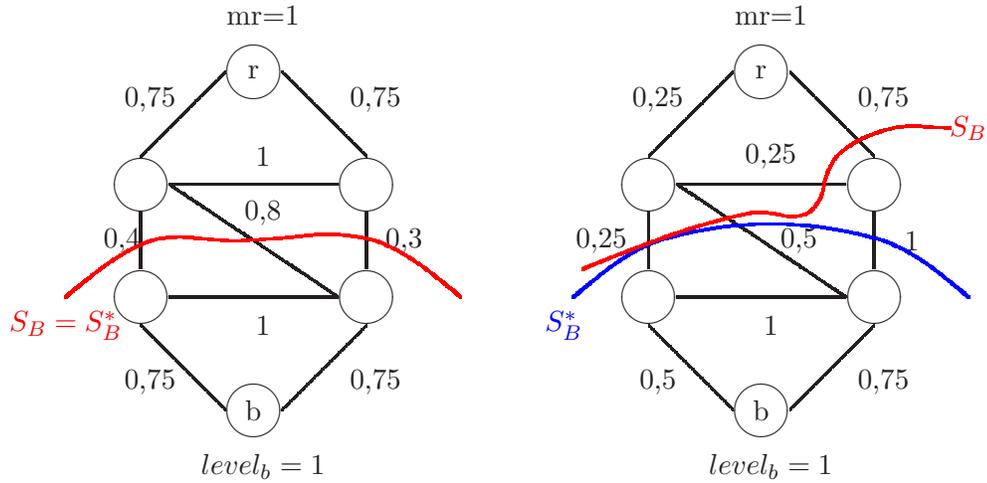

\noindent \begin{centering} \include{ExReliability}
  \par\end{centering}
 \caption{Examples of containment areas for reliability spanning tree construction.}
\label{fig:ExReliability}
\end{figure}

Now, we can state our generalization of Theorem \ref{th:impTAStrongBFS}.

\begin{theorem}\label{th:impTAstrong}
Given a maximizable metric $\mathcal{M}=(M,W,mr,met,\prec)$, even under the central daemon, there exists no $(t,A_B^*,1)$-TA-strongly stabilizing protocol for maximum metric spanning tree construction with respect to $\mathcal{M}$ where $A_B^*\varsubsetneq S_B^*$ and $t$ is a given finite integer.
\end{theorem}

\begin{proof}
Let $\mathcal{M}=(M,W,mr,met,\prec)$ be a maximizable metric and $\mathcal{P}$ be a $(t,A_B^*,1)$-TA-strongly stabilizing protocol for maximum metric spanning tree construction protocol with respect to $\mathcal{M}$ where $A_B^*\varsubsetneq S_B^*$ and $t$ is a finite integer. We must distinguish the following cases:

\begin{description}
\item[Case 1:] $|M|=1$.\\
Denote by $m$ the metric value such that $M=\{m\}$. For any system and for any process $v$, we have $\mu(v,r)=\underset{b\in B}{min_\prec}\{\mu(v,b)\}=m$. Consequently, $S_B^*=\emptyset$ for any system. Then, it is absurd to have $A_B^*\varsubsetneq S_B^*$.
\item[Case 2:] $|M|\geq 2$.\\
By definition of a bounded metric, we can deduce that there exists $m\in M$ and $w\in W$ such that $m=met(mr,w)\prec mr$. Then, we must distinguish the following cases:
\begin{description}
\item[Case 2.1:] $m$ is a fixed point of $\mathcal{M}$.\\
Let $S$ be a system such that any edge incident to the root or a Byzantine process has a weight equals to $w$. Then, we can deduce that we have: $m=\underset{b\in B}{max_\prec}\{\mu(r,b)\}\prec \mu(r,r)=mr$ and for any correct process $v\neq r$, $\mu(v,r)=\underset{b\in B}{max_\prec}\{\mu(v,b)\}=m$. Hence, $S_B^*=\emptyset$ for any such system. Then, it is absurd to have $A_B^*\varsubsetneq S_B^*$.
\item[Case 2.2:] $m$ is not a fixed point of $\mathcal{M}$.\\
This implies that there exists $w'\in W$ such that: $met(m,w')\prec m$ (remember that $\mathcal{M}$ is bounded). Consider the following system: $V=\{r,u,u',v,v',b\}$, $E=\{\{r,u\},\{r,u'\},$ $\{u,v\},\{u',v'\},\{v,b\},\{v',b\}\}$, $w_{r,u}=w_{r,u'}=w_{v,b}=w_{v',b}=w$, and $w_{u,v}=w_{u',v'}=w'$ ($b$ is a Byzantine process). We can see that $S_B^*=\{v,v'\}$. Since $A_B^*\varsubsetneq S_B$, we have: $v\notin A_B^*$ or $v'\notin A_B^*$. Consider now the following configuration $\rho_0$: $prnt_r=prnt_b=\bot$, $level_r=level_b=mr$, and $prnt$, $level$  variables of other processes are arbitrary (see Figure \ref{fig:impTAstrong}, other variables may have arbitrary values but other variables of $b$ are identical to those of $r$).

Assume now that $b$ takes exactly the same actions as $r$ (if any) immediately after $r$ (note that $r\notin A_B^*$ and hence $prnt_r=\bot$ and $level_r=mr$ still hold by closure and then $prnt_b=\bot$ and $level_b=mr$ still hold too). Then, by symmetry of the execution and by convergence of $\mathcal{P}$ to $spec$, we can deduce that the system reaches in a finite time a configuration $\rho_1$ (see Figure \ref{fig:impTAstrong}) in which: $prnt_r=prnt_b=\bot$, $prnt_u=prnt_{u'}=r$, $prnt_v=prnt_{v'}=b$, $level_r=level_b=mr$, and $level_u=level_{u'}=level_v=level_{v'}=m$ (because this configuration is the only one in which all correct process $v$ satisfies $spec(v)$ when $prnt_r=prnt_b=\bot$ and $level_r=level_b=mr$ since $met(m,w')\prec m$). Note that $\rho_1$ is $A_B^*$-legitimate for $spec$ and $A_B^*$-stable (whatever $A_B^*$ is).

Assume now that $b$ behaves as a correct processor with respect to $\mathcal{P}$. Then, by convergence of $\mathcal{P}$ in a fault-free system starting from $\rho_1$ which is not legitimate (remember that a TA-strongly stabilizing algorithm is a special case of self-stabilizing algorithm), we can deduce that the system reach in a finite time a configuration $\rho_2$ (see Figure \ref{fig:impTAstrong}) in which: $prnt_r=\bot$, $prnt_u=prnt_{u'}=r$, $prnt_v=u$, $prnt_{v'}=u'$, $prnt_b=v$ (or $prnt_b=v'$), $level_r=mr$, $level_u=level_{u'}=m$ $level_v=level_{v'}=met(m,w')=m'$, and $level_b=met(m',w)=m''$. Note that processes $v$ and $v'$ modify their O-variables in the portion of execution between $\rho_1$ and $\rho_2$ and that $\rho_2$ is $A_B^*$-legitimate for $spec$ and $A_B^*$-stable (whatever $A_B^*$ is). Consequently, this portion of execution contains at least one $A_B^*$-TA-disruption (whatever $A_B^*$ is).

Assume now that the Byzantine process $b$ takes the following state: $prnt_b=\bot$ and $level_b=mr$. This step brings the system into configuration $\rho_3$ (see Figure \ref{fig:impTAstrong}). From this configuration, we can repeat the execution we constructed from $\rho_0$. By the same token, we obtain an execution of $\mathcal{P}$ which contains $c$-legitimate and $c$-stable configurations (see $\rho_1$) and an infinite number of $A_B^*$-TA-disruption (whatever $A_B^*$ is) which contradicts the $(t,A_B^*,1)$-TA-strong stabilization of $\mathcal{P}$.
\end{description}
\end{description}
\end{proof}

\begin{figure}[t]
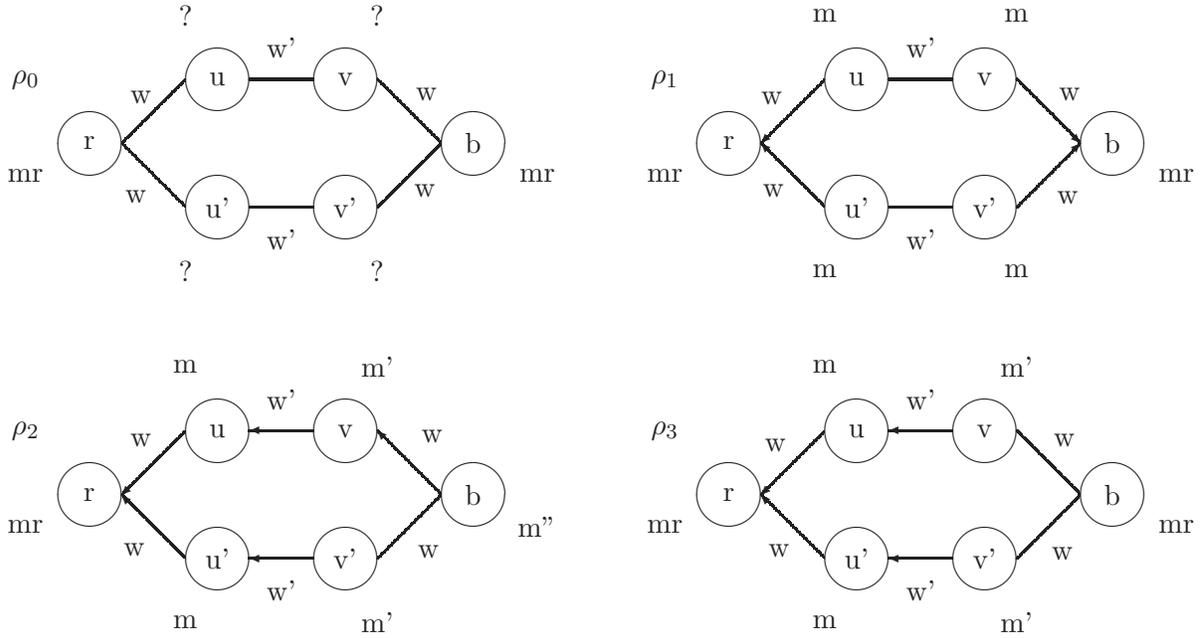

\noindent \begin{centering} \include{impTAstrong}
  \par\end{centering}
 \caption{Configurations used in proof of Theorem \ref{th:impTAstrong}.}
\label{fig:impTAstrong}
\end{figure}

\section{Conclusion}

In this paper, we presented two necessary conditions to achieve strong stabilization and topology-aware strong stabilization in maximum metric tree construction. Our work obviously leads to the following open question: is there a topology-aware strongly stabilizing protocol that ensures a containmemt area equal to $S_B^*$? We conjecture that it is the case.   

\bibliographystyle{plain}
\bibliography{biblio}

\end{document}

%% file: possStrongCase1.tex
\ifx\JPicScale\undefined\def\JPicScale{0.85}\fi
\unitlength \JPicScale mm
\begin{picture}(190,135)(0,0)
\linethickness{0.3mm}
\put(25,125){\circle{10}}

\linethickness{0.3mm}
\put(45,125){\circle{10}}

\linethickness{0.3mm}
\put(75,125){\circle{10}}

\linethickness{0.3mm}
\put(95,125){\circle{10}}

\linethickness{0.3mm}
\put(115,125){\circle{10}}

\linethickness{0.3mm}
\put(135,125){\circle{10}}

\linethickness{0.3mm}
\put(165,125){\circle{10}}

\linethickness{0.3mm}
\put(185,125){\circle{10}}

\put(10,125){\makebox(0,0)[cc]{$\rho_0$}}

\put(30,125){\makebox(0,0)[cc]{}}

\linethickness{0.3mm}
\put(120,125){\line(1,0){10}}
\linethickness{0.3mm}
\put(140,125){\line(1,0){5}}
\linethickness{0.3mm}
\put(155,125){\line(1,0){5}}
\linethickness{0.3mm}
\put(170,125){\line(1,0){10}}
\linethickness{0.3mm}
\put(100,125){\line(1,0){10}}
\linethickness{0.3mm}
\put(80,125){\line(1,0){10}}
\linethickness{0.3mm}
\put(65,125){\line(1,0){5}}
\linethickness{0.3mm}
\put(50,125){\line(1,0){5}}
\linethickness{0.3mm}
\put(30,125){\line(1,0){10}}
\put(25,135){\makebox(0,0)[cc]{$p_0=r$}}

\put(45,135){\makebox(0,0)[cc]{$p_1$}}

\put(75,135){\makebox(0,0)[cc]{$p_c$}}

\put(95,135){\makebox(0,0)[cc]{$p_{c+1}$}}

\put(115,135){\makebox(0,0)[cc]{$p_{c+2}$}}

\put(135,135){\makebox(0,0)[cc]{$p_{c+3}$}}

\put(165,135){\makebox(0,0)[cc]{$p_{2c+2}$}}

\put(60,125){\makebox(0,0)[cc]{$\ldots$}}

\put(150,125){\makebox(0,0)[cc]{$\ldots$}}

\put(60,135){\makebox(0,0)[cc]{$\ldots$}}

\put(150,135){\makebox(0,0)[cc]{$\ldots$}}

\put(185,135){\makebox(0,0)[cc]{$p_{2c+3}=b$}}

\put(35,130){\makebox(0,0)[cc]{$w_0$}}

\put(175,130){\makebox(0,0)[cc]{$w_0$}}

\put(85,130){\makebox(0,0)[cc]{$w_c$}}

\put(125,130){\makebox(0,0)[cc]{$w_c$}}

\put(105,130){\makebox(0,0)[cc]{$w_{c+1}$}}

\linethickness{0.3mm}
\put(25,90){\circle{10}}

\linethickness{0.3mm}
\put(45,90){\circle{10}}

\linethickness{0.3mm}
\put(75,90){\circle{10}}

\linethickness{0.3mm}
\put(95,90){\circle{10}}

\linethickness{0.3mm}
\put(115,90){\circle{10}}

\linethickness{0.3mm}
\put(135,90){\circle{10}}

\linethickness{0.3mm}
\put(165,90){\circle{10}}

\linethickness{0.3mm}
\put(185,90){\circle{10}}

\linethickness{0.3mm}
\put(170,90){\line(1,0){10}}
\put(180,90){\vector(1,0){0.12}}
\linethickness{0.3mm}
\put(120,90){\line(1,0){10}}
\put(130,90){\vector(1,0){0.12}}
\linethickness{0.3mm}
\put(80,90){\line(1,0){10}}
\put(80,90){\vector(-1,0){0.12}}
\linethickness{0.3mm}
\put(65,90){\line(1,0){5}}
\put(65,90){\vector(-1,0){0.12}}
\linethickness{0.3mm}
\put(30,90){\line(1,0){10}}
\put(30,90){\vector(-1,0){0.12}}
\linethickness{0.3mm}
\put(50,90){\line(1,0){5}}
\put(50,90){\vector(-1,0){0.12}}
\linethickness{0.3mm}
\put(155,90){\line(1,0){5}}
\put(160,90){\vector(1,0){0.12}}
\linethickness{0.3mm}
\put(100,90){\line(1,0){10}}
\put(10,90){\makebox(0,0)[cc]{$\rho_1$}}

\put(60,90){\makebox(0,0)[cc]{$\ldots$}}

\put(150,90){\makebox(0,0)[cc]{$\ldots$}}

\linethickness{0.3mm}
\put(140,90){\line(1,0){5}}
\put(145,90){\vector(1,0){0.12}}
\linethickness{0.3mm}
\put(25,55){\circle{10}}

\linethickness{0.3mm}
\put(45,55){\circle{10}}

\linethickness{0.3mm}
\put(75,55){\circle{10}}

\linethickness{0.3mm}
\put(95,55){\circle{10}}

\linethickness{0.3mm}
\put(115,55){\circle{10}}

\linethickness{0.3mm}
\put(135,55){\circle{10}}

\linethickness{0.3mm}
\put(165,55){\circle{10}}

\linethickness{0.3mm}
\put(185,55){\circle{10}}

\linethickness{0.3mm}
\put(30,55){\line(1,0){10}}
\put(30,55){\vector(-1,0){0.12}}
\linethickness{0.3mm}
\put(80,55){\line(1,0){10}}
\put(80,55){\vector(-1,0){0.12}}
\linethickness{0.3mm}
\put(100,55){\line(1,0){10}}
\put(100,55){\vector(-1,0){0.12}}
\linethickness{0.3mm}
\put(120,55){\line(1,0){10}}
\put(120,55){\vector(-1,0){0.12}}
\linethickness{0.3mm}
\put(170,55){\line(1,0){10}}
\put(170,55){\vector(-1,0){0.12}}
\linethickness{0.3mm}
\put(155,55){\line(1,0){5}}
\put(155,55){\vector(-1,0){0.12}}
\linethickness{0.3mm}
\put(50,55){\line(1,0){5}}
\put(50,55){\vector(-1,0){0.12}}
\linethickness{0.3mm}
\put(140,55){\line(1,0){5}}
\put(140,55){\vector(-1,0){0.12}}
\linethickness{0.3mm}
\put(65,55){\line(1,0){5}}
\put(65,55){\vector(-1,0){0.12}}
\put(10,55){\makebox(0,0)[cc]{$\rho_2$}}

\put(60,55){\makebox(0,0)[cc]{$\ldots$}}

\put(150,55){\makebox(0,0)[cc]{$\ldots$}}

\linethickness{0.3mm}
\put(25,20){\circle{10}}

\linethickness{0.3mm}
\put(45,20){\circle{10}}

\linethickness{0.3mm}
\put(75,20){\circle{10}}

\linethickness{0.3mm}
\put(95,20){\circle{10}}

\linethickness{0.3mm}
\put(115,20){\circle{10}}

\linethickness{0.3mm}
\put(135,20){\circle{10}}

\linethickness{0.3mm}
\put(165,20){\circle{10}}

\linethickness{0.3mm}
\put(185,20){\circle{10}}

\linethickness{0.3mm}
\put(30,20){\line(1,0){10}}
\put(30,20){\vector(-1,0){0.12}}
\linethickness{0.3mm}
\put(80,20){\line(1,0){10}}
\put(80,20){\vector(-1,0){0.12}}
\linethickness{0.3mm}
\put(100,20){\line(1,0){10}}
\put(100,20){\vector(-1,0){0.12}}
\linethickness{0.3mm}
\put(120,20){\line(1,0){10}}
\put(120,20){\vector(-1,0){0.12}}
\linethickness{0.3mm}
\put(170,20){\line(1,0){10}}
\put(170,20){\vector(-1,0){0.12}}
\linethickness{0.3mm}
\put(155,20){\line(1,0){5}}
\put(155,20){\vector(-1,0){0.12}}
\linethickness{0.3mm}
\put(65,20){\line(1,0){5}}
\put(65,20){\vector(-1,0){0.12}}
\linethickness{0.3mm}
\put(50,20){\line(1,0){5}}
\put(50,20){\vector(-1,0){0.12}}
\linethickness{0.3mm}
\put(140,20){\line(1,0){5}}
\put(140,20){\vector(-1,0){0.12}}
\put(10,20){\makebox(0,0)[cc]{$\rho_3$}}

\put(60,20){\makebox(0,0)[cc]{$\ldots$}}

\put(150,20){\makebox(0,0)[cc]{$\ldots$}}

\put(25,115){\makebox(0,0)[cc]{mr}}

\put(185,115){\makebox(0,0)[cc]{mr}}

\put(45,115){\makebox(0,0)[cc]{?}}

\put(75,115){\makebox(0,0)[cc]{?}}

\put(95,115){\makebox(0,0)[cc]{?}}

\put(115,115){\makebox(0,0)[cc]{?}}

\put(135,115){\makebox(0,0)[cc]{?}}

\put(165,115){\makebox(0,0)[cc]{?}}

\put(25,80){\makebox(0,0)[cc]{mr}}

\put(185,100){\makebox(0,0)[cc]{mr}}

\put(45,100){\makebox(0,0)[cc]{$\mu(p_{1},r)$}}

\put(75,80){\makebox(0,0)[cc]{$\mu(p_{c},r)$}}

\put(95,100){\makebox(0,0)[cc]{$\mu(p_{c+1},r)$}}

\put(115,80){\makebox(0,0)[cc]{$\mu(p_{c+2},b)$}}

\put(135,100){\makebox(0,0)[cc]{$\mu(p_{c+3},b)$}}

\put(165,80){\makebox(0,0)[cc]{$\mu(p_{2c+2},b)$}}

\put(25,45){\makebox(0,0)[cc]{mr}}

\put(45,65){\makebox(0,0)[cc]{$\mu(p_{1},r)$}}

\put(75,45){\makebox(0,0)[cc]{$\mu(p_{c},r)$}}

\put(95,65){\makebox(0,0)[cc]{$\mu(p_{c+1},r)$}}

\put(115,45){\makebox(0,0)[cc]{$\mu(p_{c+2},r)$}}

\put(135,65){\makebox(0,0)[cc]{$\mu(p_{c+3},r)$}}

\put(165,45){\makebox(0,0)[cc]{$\mu(p_{2c+2},r)$}}

\put(185,65){\makebox(0,0)[cc]{$\mu(p_{b},r)$}}

\put(25,10){\makebox(0,0)[cc]{mr}}

\put(45,30){\makebox(0,0)[cc]{$\mu(p_{1},r)$}}

\put(75,10){\makebox(0,0)[cc]{$\mu(p_{c},r)$}}

\put(95,30){\makebox(0,0)[cc]{$\mu(p_{c+1},r)$}}

\put(115,10){\makebox(0,0)[cc]{$\mu(p_{c+2},r)$}}

\put(135,30){\makebox(0,0)[cc]{$\mu(p_{c+3},r)$}}

\put(165,10){\makebox(0,0)[cc]{$\mu(p_{2c+2},r)$}}

\put(185,30){\makebox(0,0)[cc]{mr}}

\end{picture}

%% file: possStrongCase23.tex
\ifx\JPicScale\undefined\def\JPicScale{0.85}\fi
\unitlength \JPicScale mm
\begin{picture}(185,100)(0,0)
\put(5,80){\makebox(0,0)[cc]{$S_1$}}

\linethickness{0.3mm}
\put(20,90){\circle{10}}

\linethickness{0.3mm}
\put(40,90){\circle{10}}

\linethickness{0.3mm}
\put(70,90){\circle{10}}

\linethickness{0.3mm}
\put(90,90){\circle{10}}

\linethickness{0.3mm}
\put(110,90){\circle{10}}

\linethickness{0.3mm}
\put(130,90){\circle{10}}

\linethickness{0.3mm}
\put(40,70){\circle{10}}

\linethickness{0.3mm}
\put(20,70){\circle{10}}

\put(25,90){\makebox(0,0)[cc]{}}

\linethickness{0.3mm}
\put(115,90){\line(1,0){10}}
\linethickness{0.3mm}
\put(60,70){\line(1,0){5}}
\linethickness{0.3mm}
\put(45,70){\line(1,0){5}}
\linethickness{0.3mm}
\put(25,70){\line(1,0){10}}
\linethickness{0.3mm}
\put(95,90){\line(1,0){10}}
\linethickness{0.3mm}
\put(75,90){\line(1,0){10}}
\linethickness{0.3mm}
\put(60,90){\line(1,0){5}}
\linethickness{0.3mm}
\put(45,90){\line(1,0){5}}
\linethickness{0.3mm}
\put(25,90){\line(1,0){10}}
\put(20,100){\makebox(0,0)[cc]{$p_0=r$}}

\put(40,100){\makebox(0,0)[cc]{$p_1$}}

\put(160,100){\makebox(0,0)[cc]{$p_c$}}

\put(180,100){\makebox(0,0)[cc]{$p_{c+1}$}}

\put(180,60){\makebox(0,0)[cc]{$p_{c+2}$}}

\put(160,60){\makebox(0,0)[cc]{$p_{c+3}$}}

\put(40,60){\makebox(0,0)[cc]{$p_{2c+2}$}}

\put(55,90){\makebox(0,0)[cc]{$\ldots$}}

\put(55,70){\makebox(0,0)[cc]{$\ldots$}}

\put(55,100){\makebox(0,0)[cc]{$\ldots$}}

\put(145,100){\makebox(0,0)[cc]{$\ldots$}}

\put(20,60){\makebox(0,0)[cc]{$p_{2c+3}=b$}}

\put(30,85){\makebox(0,0)[cc]{$w_0$}}

\put(30,75){\makebox(0,0)[cc]{$w_0$}}

\linethickness{0.3mm}
\put(70,70){\circle{10}}

\linethickness{0.3mm}
\put(90,70){\circle{10}}

\linethickness{0.3mm}
\put(110,70){\circle{10}}

\linethickness{0.3mm}
\put(130,70){\circle{10}}

\linethickness{0.3mm}
\put(180,90){\circle{10}}

\linethickness{0.3mm}
\put(180,70){\circle{10}}

\linethickness{0.3mm}
\put(160,70){\circle{10}}

\linethickness{0.3mm}
\put(160,90){\circle{10}}

\linethickness{0.3mm}
\put(150,90){\line(1,0){5}}
\linethickness{0.3mm}
\put(135,90){\line(1,0){5}}
\put(145,90){\makebox(0,0)[cc]{$\ldots$}}

\linethickness{0.3mm}
\put(150,70){\line(1,0){5}}
\linethickness{0.3mm}
\put(135,70){\line(1,0){5}}
\put(145,70){\makebox(0,0)[cc]{$\ldots$}}

\linethickness{0.3mm}
\put(165,90){\line(1,0){10}}
\linethickness{0.3mm}
\put(180,75){\line(0,1){10}}
\linethickness{0.3mm}
\put(165,70){\line(1,0){10}}
\linethickness{0.3mm}
\put(115,70){\line(1,0){10}}
\linethickness{0.3mm}
\put(95,70){\line(1,0){10}}
\linethickness{0.3mm}
\put(75,70){\line(1,0){10}}
\put(90,100){\makebox(0,0)[cc]{$p_k$}}

\put(70,100){\makebox(0,0)[cc]{$p_{k-1}$}}

\put(110,100){\makebox(0,0)[cc]{$p_{k+1}$}}

\put(130,100){\makebox(0,0)[cc]{$p_{k+2}$}}

\put(185,80){\makebox(0,0)[cc]{$w'$}}

\put(170,85){\makebox(0,0)[cc]{$w$}}

\put(170,75){\makebox(0,0)[cc]{$w$}}

\put(120,85){\makebox(0,0)[cc]{$w$}}

\put(120,75){\makebox(0,0)[cc]{$w$}}

\put(100,85){\makebox(0,0)[cc]{$w$}}

\put(80,85){\makebox(0,0)[cc]{$w_{k-1}$}}

\put(80,75){\makebox(0,0)[cc]{$w_{k-1}$}}

\put(90,60){\makebox(0,0)[cc]{$p_{2c+3-k}$}}

\put(110,60){\makebox(0,0)[cc]{$p_{2c+2-k}$}}

\put(70,60){\makebox(0,0)[cc]{$p_{2c+4-k}$}}

\put(130,60){\makebox(0,0)[cc]{$p_{2c+1-k}$}}

\linethickness{0.3mm}
\put(20,35){\circle{10}}

\linethickness{0.3mm}
\put(40,35){\circle{10}}

\linethickness{0.3mm}
\put(70,35){\circle{10}}

\linethickness{0.3mm}
\put(90,35){\circle{10}}

\linethickness{0.3mm}
\put(110,35){\circle{10}}

\linethickness{0.3mm}
\put(130,35){\circle{10}}

\linethickness{0.3mm}
\put(40,15){\circle{10}}

\linethickness{0.3mm}
\put(20,15){\circle{10}}

\put(25,35){\makebox(0,0)[cc]{}}

\linethickness{0.3mm}
\put(115,35){\line(1,0){10}}
\linethickness{0.3mm}
\put(60,15){\line(1,0){5}}
\linethickness{0.3mm}
\put(45,15){\line(1,0){5}}
\linethickness{0.3mm}
\put(25,15){\line(1,0){10}}
\linethickness{0.3mm}
\put(95,35){\line(1,0){10}}
\linethickness{0.3mm}
\put(75,35){\line(1,0){10}}
\linethickness{0.3mm}
\put(60,35){\line(1,0){5}}
\linethickness{0.3mm}
\put(45,35){\line(1,0){5}}
\linethickness{0.3mm}
\put(25,35){\line(1,0){10}}
\put(20,45){\makebox(0,0)[cc]{$p_0=r$}}

\put(40,45){\makebox(0,0)[cc]{$p_1$}}

\put(160,45){\makebox(0,0)[cc]{$p_c$}}

\put(180,45){\makebox(0,0)[cc]{$p_{c+1}$}}

\put(180,5){\makebox(0,0)[cc]{$p_{c+2}$}}

\put(160,5){\makebox(0,0)[cc]{$p_{c+3}$}}

\put(40,5){\makebox(0,0)[cc]{$p_{2c+2}$}}

\put(55,35){\makebox(0,0)[cc]{$\ldots$}}

\put(55,15){\makebox(0,0)[cc]{$\ldots$}}

\put(55,45){\makebox(0,0)[cc]{$\ldots$}}

\put(145,45){\makebox(0,0)[cc]{$\ldots$}}

\put(20,5){\makebox(0,0)[cc]{$p_{2c+3}=b$}}

\put(30,30){\makebox(0,0)[cc]{$w_0$}}

\linethickness{0.3mm}
\put(70,15){\circle{10}}

\linethickness{0.3mm}
\put(90,15){\circle{10}}

\linethickness{0.3mm}
\put(110,15){\circle{10}}

\linethickness{0.3mm}
\put(130,15){\circle{10}}

\linethickness{0.3mm}
\put(180,35){\circle{10}}

\linethickness{0.3mm}
\put(180,15){\circle{10}}

\linethickness{0.3mm}
\put(160,15){\circle{10}}

\linethickness{0.3mm}
\put(160,35){\circle{10}}

\linethickness{0.3mm}
\put(150,35){\line(1,0){5}}
\linethickness{0.3mm}
\put(135,35){\line(1,0){5}}
\put(145,35){\makebox(0,0)[cc]{$\ldots$}}

\linethickness{0.3mm}
\put(150,15){\line(1,0){5}}
\linethickness{0.3mm}
\put(135,15){\line(1,0){5}}
\put(145,15){\makebox(0,0)[cc]{$\ldots$}}

\linethickness{0.3mm}
\put(165,35){\line(1,0){10}}
\linethickness{0.3mm}
\put(180,20){\line(0,1){10}}
\linethickness{0.3mm}
\put(165,15){\line(1,0){10}}
\linethickness{0.3mm}
\put(115,15){\line(1,0){10}}
\linethickness{0.3mm}
\put(95,15){\line(1,0){10}}
\linethickness{0.3mm}
\put(75,15){\line(1,0){10}}
\put(90,45){\makebox(0,0)[cc]{$p_k$}}

\put(70,45){\makebox(0,0)[cc]{$p_{k-1}$}}

\put(110,45){\makebox(0,0)[cc]{$p_{k+1}$}}

\put(130,45){\makebox(0,0)[cc]{$p_{k+2}$}}

\put(170,30){\makebox(0,0)[cc]{$w$}}

\put(170,20){\makebox(0,0)[cc]{$w$}}

\put(120,30){\makebox(0,0)[cc]{$w$}}

\put(120,20){\makebox(0,0)[cc]{$w$}}

\put(100,30){\makebox(0,0)[cc]{$w$}}

\put(80,30){\makebox(0,0)[cc]{$w_{k-1}$}}

\put(5,25){\makebox(0,0)[cc]{$S_2$}}

\put(185,25){\makebox(0,0)[cc]{$w$}}

\put(30,20){\makebox(0,0)[cc]{$w'_0$}}

\put(80,20){\makebox(0,0)[cc]{$w'_{k'-1}$}}

\put(90,5){\makebox(0,0)[cc]{$p_{2c+3-k'}$}}

\put(110,5){\makebox(0,0)[cc]{$p_{2c+2-k'}$}}

\put(70,5){\makebox(0,0)[cc]{$p_{2c+4-k'}$}}

\put(130,5){\makebox(0,0)[cc]{$p_{2c+1-k'}$}}

\put(145,60){\makebox(0,0)[cc]{$\ldots$}}

\put(55,60){\makebox(0,0)[cc]{$\ldots$}}

\put(145,5){\makebox(0,0)[cc]{$\ldots$}}

\put(55,5){\makebox(0,0)[cc]{$\ldots$}}

\put(100,20){\makebox(0,0)[cc]{$w$}}

\put(100,75){\makebox(0,0)[cc]{$w$}}

\end{picture}

%% file: ExSP.tex
\ifx\JPicScale\undefined\def\JPicScale{0.75}\fi
\unitlength \JPicScale mm
\begin{picture}(165,95)(0,0)
\linethickness{0.3mm}
\put(50,85){\circle{10}}

\linethickness{0.3mm}
\put(30,65){\circle{10}}

\linethickness{0.3mm}
\put(70,65){\circle{10}}

\linethickness{0.3mm}
\put(30,45){\circle{10}}

\linethickness{0.3mm}
\put(70,45){\circle{10}}

\linethickness{0.3mm}
\put(50,25){\circle{10}}

\linethickness{0.3mm}
\multiput(30,70)(0.12,0.12){125}{\line(1,0){0.12}}
\linethickness{0.3mm}
\multiput(55,85)(0.12,-0.12){125}{\line(1,0){0.12}}
\linethickness{0.3mm}
\put(30,50){\line(0,1){10}}
\linethickness{0.3mm}
\multiput(30,40)(0.12,-0.12){125}{\line(1,0){0.12}}
\linethickness{0.3mm}
\multiput(55,25)(0.12,0.12){125}{\line(1,0){0.12}}
\linethickness{0.3mm}
\put(70,50){\line(0,1){10}}
\linethickness{0.3mm}
\put(35,65){\line(1,0){30}}
\linethickness{0.3mm}
\put(35,45){\line(1,0){30}}
\linethickness{0.3mm}
\multiput(35,65)(0.18,-0.12){167}{\line(1,0){0.18}}
\linethickness{0.3mm}
\put(140,85){\circle{10}}

\linethickness{0.3mm}
\put(120,65){\circle{10}}

\linethickness{0.3mm}
\put(160,65){\circle{10}}

\linethickness{0.3mm}
\put(120,45){\circle{10}}

\linethickness{0.3mm}
\put(160,45){\circle{10}}

\linethickness{0.3mm}
\put(140,25){\circle{10}}

\linethickness{0.3mm}
\multiput(120,70)(0.12,0.12){125}{\line(1,0){0.12}}
\linethickness{0.3mm}
\multiput(145,85)(0.12,-0.12){125}{\line(1,0){0.12}}
\linethickness{0.3mm}
\put(120,50){\line(0,1){10}}
\linethickness{0.3mm}
\multiput(120,40)(0.12,-0.12){125}{\line(1,0){0.12}}
\linethickness{0.3mm}
\multiput(145,25)(0.12,0.12){125}{\line(1,0){0.12}}
\linethickness{0.3mm}
\put(160,50){\line(0,1){10}}
\linethickness{0.3mm}
\put(125,65){\line(1,0){30}}
\linethickness{0.3mm}
\put(125,45){\line(1,0){30}}
\linethickness{0.3mm}
\multiput(125,65)(0.18,-0.12){167}{\line(1,0){0.18}}
\put(50,85){\makebox(0,0)[cc]{r}}

\put(140,85){\makebox(0,0)[cc]{r}}

\put(140,25){\makebox(0,0)[cc]{b}}

\put(50,25){\makebox(0,0)[cc]{b}}

\textcolor{blue}{
\put(115,25){\makebox(0,0)[cc]{$S_B^*$}}
}

\textcolor{red}{
\put(165,75){\makebox(0,0)[cc]{$S_B$}}
\put(20,25){\makebox(0,0)[cc]{$S_B=S_B^*$}}
}

\put(50,95){\makebox(0,0)[cc]{mr=0}}

\put(140,95){\makebox(0,0)[cc]{mr=0}}

\put(50,15){\makebox(0,0)[cc]{$level_b=0$}}

\put(140,15){\makebox(0,0)[cc]{$level_b=0$}}

\put(35,80){\makebox(0,0)[cc]{7}}

\put(65,80){\makebox(0,0)[cc]{6}}

\put(50,70){\makebox(0,0)[cc]{5}}

\put(75,55){\makebox(0,0)[cc]{4}}

\put(55,55){\makebox(0,0)[cc]{10}}

\put(50,40){\makebox(0,0)[cc]{8}}

\put(25,55){\makebox(0,0)[cc]{6}}

\put(65,30){\makebox(0,0)[cc]{32}}

\put(35,30){\makebox(0,0)[cc]{16}}

\textcolor{red}{
\linethickness{0.3mm}
\qbezier(25,30)(25.43,30.23)(30.48,32.17)
\qbezier(30.48,32.17)(35.53,34.11)(40,35)
\qbezier(40,35)(42.53,35.36)(44.98,35.21)
\qbezier(44.98,35.21)(47.43,35.05)(50,35)
\qbezier(50,35)(52.57,35.05)(55.02,35.21)
\qbezier(55.02,35.21)(57.47,35.36)(60,35)
\qbezier(60,35)(64.47,34.11)(69.52,32.17)
\qbezier(69.52,32.17)(74.57,30.23)(75,30)
}

\put(125,80){\makebox(0,0)[cc]{0}}

\put(155,80){\makebox(0,0)[cc]{0}}

\put(140,70){\makebox(0,0)[cc]{0}}

\put(165,55){\makebox(0,0)[cc]{0}}

\put(115,55){\makebox(0,0)[cc]{0}}

\put(155,30){\makebox(0,0)[cc]{0}}

\put(125,30){\makebox(0,0)[cc]{0}}

\put(140,40){\makebox(0,0)[cc]{0}}

\put(145,55){\makebox(0,0)[cc]{0}}

\textcolor{blue}{
\linethickness{0.3mm}
\qbezier(115,29.45)(115.43,29.68)(120.48,31.62)
\qbezier(120.48,31.62)(125.53,33.56)(130,34.45)
\qbezier(130,34.45)(132.53,34.81)(134.98,34.65)
\qbezier(134.98,34.65)(137.43,34.5)(140,34.45)
\qbezier(140,34.45)(142.57,34.5)(145.02,34.65)
\qbezier(145.02,34.65)(147.47,34.81)(150,34.45)
\qbezier(150,34.45)(154.47,33.56)(159.52,31.62)
\qbezier(159.52,31.62)(164.57,29.68)(165,29.45)
}

\textcolor{red}{
\linethickness{0.3mm}
\qbezier(115,80)(115.43,79.77)(120.48,77.83)
\qbezier(120.48,77.83)(125.53,75.89)(130,75)
\qbezier(130,75)(135.08,74.17)(140,74.17)
\qbezier(140,74.17)(144.92,74.17)(150,75)
\qbezier(150,75)(153.99,75.69)(157.63,77.04)
\qbezier(157.63,77.04)(161.26,78.4)(165,80)
}
\end{picture}

%% file: ExFlow.tex
\ifx\JPicScale\undefined\def\JPicScale{0.75}\fi
\unitlength \JPicScale mm
\begin{picture}(173,95)(0,0)
\linethickness{0.3mm}
\put(50,85){\circle{10}}

\linethickness{0.3mm}
\put(30,65){\circle{10}}

\linethickness{0.3mm}
\put(70,65){\circle{10}}

\linethickness{0.3mm}
\put(30,45){\circle{10}}

\linethickness{0.3mm}
\put(70,45){\circle{10}}

\linethickness{0.3mm}
\put(50,25){\circle{10}}

\linethickness{0.3mm}
\multiput(30,70)(0.12,0.12){125}{\line(1,0){0.12}}
\linethickness{0.3mm}
\multiput(55,85)(0.12,-0.12){125}{\line(1,0){0.12}}
\linethickness{0.3mm}
\put(30,50){\line(0,1){10}}
\linethickness{0.3mm}
\multiput(30,40)(0.12,-0.12){125}{\line(1,0){0.12}}
\linethickness{0.3mm}
\multiput(55,25)(0.12,0.12){125}{\line(1,0){0.12}}
\linethickness{0.3mm}
\put(70,50){\line(0,1){10}}
\linethickness{0.3mm}
\put(35,65){\line(1,0){30}}
\linethickness{0.3mm}
\put(35,45){\line(1,0){30}}
\linethickness{0.3mm}
\multiput(35,65)(0.18,-0.12){167}{\line(1,0){0.18}}
\linethickness{0.3mm}
\put(140,85){\circle{10}}

\linethickness{0.3mm}
\put(120,65){\circle{10}}

\linethickness{0.3mm}
\put(160,65){\circle{10}}

\linethickness{0.3mm}
\put(120,45){\circle{10}}

\linethickness{0.3mm}
\put(160,45){\circle{10}}

\linethickness{0.3mm}
\put(140,25){\circle{10}}

\linethickness{0.3mm}
\multiput(120,70)(0.12,0.12){125}{\line(1,0){0.12}}
\linethickness{0.3mm}
\multiput(145,85)(0.12,-0.12){125}{\line(1,0){0.12}}
\linethickness{0.3mm}
\put(120,50){\line(0,1){10}}
\linethickness{0.3mm}
\multiput(120,40)(0.12,-0.12){125}{\line(1,0){0.12}}
\linethickness{0.3mm}
\multiput(145,25)(0.12,0.12){125}{\line(1,0){0.12}}
\linethickness{0.3mm}
\put(160,50){\line(0,1){10}}
\linethickness{0.3mm}
\put(125,65){\line(1,0){30}}
\linethickness{0.3mm}
\put(125,45){\line(1,0){30}}
\linethickness{0.3mm}
\multiput(125,65)(0.18,-0.12){167}{\line(1,0){0.18}}
\put(50,85){\makebox(0,0)[cc]{r}}

\put(140,85){\makebox(0,0)[cc]{r}}

\put(140,25){\makebox(0,0)[cc]{b}}

\put(50,25){\makebox(0,0)[cc]{b}}

\put(50,95){\makebox(0,0)[cc]{mr=10}}

\put(140,95){\makebox(0,0)[cc]{mr=10}}

\put(35,80){\makebox(0,0)[cc]{7}}

\put(65,80){\makebox(0,0)[cc]{6}}

\put(50,70){\makebox(0,0)[cc]{5}}

\put(75,52){\makebox(0,0)[cc]{4}}

\put(55,55){\makebox(0,0)[cc]{10}}

\put(25,55){\makebox(0,0)[cc]{6}}

\put(50,40){\makebox(0,0)[cc]{8}}

\put(65,30){\makebox(0,0)[cc]{32}}

\put(30,30){\makebox(0,0)[cc]{16}}

\put(140,15){\makebox(0,0)[cc]{$level_b=10$}}

\put(50,15){\makebox(0,0)[cc]{$level_b=10$}}

\put(120,30){\makebox(0,0)[cc]{11}}

\put(155,30){\makebox(0,0)[cc]{12}}

\put(155,80){\makebox(0,0)[cc]{10}}

\put(120,80){\makebox(0,0)[cc]{7}}

\put(165,55){\makebox(0,0)[cc]{13}}

\put(140,70){\makebox(0,0)[cc]{6}}

\put(145,55){\makebox(0,0)[cc]{5}}

\put(115,55){\makebox(0,0)[cc]{3}}

\put(140,40){\makebox(0,0)[cc]{1}}

\textcolor{blue}{
\linethickness{0.3mm}
\qbezier(105,48.29)(105.59,48.67)(112.37,51.07)
\qbezier(112.37,51.07)(119.14,53.48)(125,53.29)
\qbezier(125,53.29)(127.81,52.88)(130.18,51.35)
\qbezier(130.18,51.35)(132.54,49.83)(135,48.29)
\qbezier(135,48.29)(138.94,45.85)(142.39,43.06)
\qbezier(142.39,43.06)(145.84,40.27)(150,38.29)
\qbezier(150,38.29)(154.85,36.17)(159.79,35.03)
\qbezier(159.79,35.03)(164.73,33.9)(170,33.29)
\put(105,44.14){\makebox(0,0)[cc]{$S_B^*$}}
}

\textcolor{red}{
\linethickness{0.3mm}
\qbezier(105,53.86)(105.43,54.17)(110.46,56.36)
\qbezier(110.46,56.36)(115.5,58.54)(120,58.86)
\qbezier(120,58.86)(124.01,58.43)(127.58,55.89)
\qbezier(127.58,55.89)(131.14,53.36)(135,53.86)
\qbezier(135,53.86)(139.98,55.48)(142.9,60.37)
\qbezier(142.9,60.37)(145.83,65.27)(150,68.86)
\qbezier(150,68.86)(152.33,70.57)(154.77,71.87)
\qbezier(154.77,71.87)(157.21,73.18)(160,73.86)
\qbezier(160,73.86)(162.53,74.5)(165.07,74.64)
\qbezier(165.07,74.64)(167.61,74.79)(170,73.86)
\put(173,68.71){\makebox(0,0)[cc]{$S_B$}}
}

\textcolor{red}{
\linethickness{0.3mm}
\qbezier(20,75)(20.66,75.25)(27.63,75.98)
\qbezier(27.63,75.98)(34.6,76.72)(40,75)
\qbezier(40,75)(41.65,74.26)(42.8,72.89)
\qbezier(42.8,72.89)(43.96,71.52)(45,70)
\qbezier(45,70)(46.51,67.55)(47.23,64.72)
\qbezier(47.23,64.72)(47.94,61.9)(50,60)
\qbezier(50,60)(54.2,56.72)(59.43,55.56)
\qbezier(59.43,55.56)(64.65,54.39)(70,55)
\qbezier(70,55)(72.83,55.38)(75.26,56.81)
\qbezier(75.26,56.81)(77.68,58.25)(80,60)
\put(10,70){\makebox(0,0)[cc]{$S_B=S_B^*$}}
}
\end{picture}

%% file: ExReliability.tex
\ifx\JPicScale\undefined\def\JPicScale{0.75}\fi
\unitlength \JPicScale mm
\begin{picture}(175,95)(0,0)
\linethickness{0.3mm}
\put(50,85){\circle{10}}

\linethickness{0.3mm}
\put(30,65){\circle{10}}

\linethickness{0.3mm}
\put(70,65){\circle{10}}

\linethickness{0.3mm}
\put(30,45){\circle{10}}

\linethickness{0.3mm}
\put(70,45){\circle{10}}

\linethickness{0.3mm}
\put(50,25){\circle{10}}

\linethickness{0.3mm}
\multiput(30,70)(0.12,0.12){125}{\line(1,0){0.12}}
\linethickness{0.3mm}
\multiput(55,85)(0.12,-0.12){125}{\line(1,0){0.12}}
\linethickness{0.3mm}
\put(30,50){\line(0,1){10}}
\linethickness{0.3mm}
\multiput(30,40)(0.12,-0.12){125}{\line(1,0){0.12}}
\linethickness{0.3mm}
\multiput(55,25)(0.12,0.12){125}{\line(1,0){0.12}}
\linethickness{0.3mm}
\put(70,50){\line(0,1){10}}
\linethickness{0.3mm}
\put(35,65){\line(1,0){30}}
\linethickness{0.3mm}
\put(35,45){\line(1,0){30}}
\linethickness{0.3mm}
\multiput(35,65)(0.18,-0.12){167}{\line(1,0){0.18}}
\linethickness{0.3mm}
\put(140,85){\circle{10}}

\linethickness{0.3mm}
\put(120,65){\circle{10}}

\linethickness{0.3mm}
\put(160,65){\circle{10}}

\linethickness{0.3mm}
\put(120,45){\circle{10}}

\linethickness{0.3mm}
\put(160,45){\circle{10}}

\linethickness{0.3mm}
\put(140,25){\circle{10}}

\linethickness{0.3mm}
\multiput(120,70)(0.12,0.12){125}{\line(1,0){0.12}}
\linethickness{0.3mm}
\multiput(145,85)(0.12,-0.12){125}{\line(1,0){0.12}}
\linethickness{0.3mm}
\put(120,50){\line(0,1){10}}
\linethickness{0.3mm}
\multiput(120,40)(0.12,-0.12){125}{\line(1,0){0.12}}
\linethickness{0.3mm}
\multiput(145,25)(0.12,0.12){125}{\line(1,0){0.12}}
\linethickness{0.3mm}
\put(160,50){\line(0,1){10}}
\linethickness{0.3mm}
\put(125,65){\line(1,0){30}}
\linethickness{0.3mm}
\put(125,45){\line(1,0){30}}
\linethickness{0.3mm}
\multiput(125,65)(0.18,-0.12){167}{\line(1,0){0.18}}
\put(50,85){\makebox(0,0)[cc]{r}}

\put(140,85){\makebox(0,0)[cc]{r}}

\put(140,25){\makebox(0,0)[cc]{b}}

\put(50,25){\makebox(0,0)[cc]{b}}

\textcolor{blue}{
\put(105,40){\makebox(0,0)[cc]{$S_B^*$}}
}

\textcolor{red}{
\put(175,75){\makebox(0,0)[cc]{$S_B$}}
\put(15,40){\makebox(0,0)[cc]{$S_B=S_B^*$}}
}

\put(50,95){\makebox(0,0)[cc]{mr=1}}

\put(50,15){\makebox(0,0)[cc]{$level_b=1$}}

\put(140,95){\makebox(0,0)[cc]{mr=1}}

\put(140,15){\makebox(0,0)[cc]{$level_b=1$}}

\put(70,80){\makebox(0,0)[cc]{0,75}}

\put(30,80){\makebox(0,0)[cc]{0,75}}

\put(70,30){\makebox(0,0)[cc]{0,75}}

\put(30,30){\makebox(0,0)[cc]{0,75}}

\put(50,40){\makebox(0,0)[cc]{1}}

\put(50,70){\makebox(0,0)[cc]{1}}

\put(50,60){\makebox(0,0)[cc]{0,8}}

\put(25,55){\makebox(0,0)[cc]{0,4}}

\put(75,55){\makebox(0,0)[cc]{0,3}}

\textcolor{red}{
\linethickness{0.3mm}
\qbezier(15,45)(15.38,45.46)(20.28,49.33)
\qbezier(20.28,49.33)(25.17,53.2)(30,55)
\qbezier(30,55)(33.69,55.98)(37.41,55.56)
\qbezier(37.41,55.56)(41.13,55.14)(45,55)
\qbezier(45,55)(51.45,55.23)(57.65,55.93)
\qbezier(57.65,55.93)(63.85,56.63)(70,55)
\qbezier(70,55)(74.83,53.2)(79.72,49.33)
\qbezier(79.72,49.33)(84.62,45.46)(85,45)
}

\put(120,80){\makebox(0,0)[cc]{0,25}}

\put(140,70){\makebox(0,0)[cc]{0,25}}

\put(160,80){\makebox(0,0)[cc]{0,75}}

\put(165,55){\makebox(0,0)[cc]{1}}

\put(145,55){\makebox(0,0)[cc]{0,5}}

\put(140,40){\makebox(0,0)[cc]{1}}

\put(110,55){\makebox(0,0)[cc]{0,25}}

\put(160,30){\makebox(0,0)[cc]{0,75}}

\put(120,30){\makebox(0,0)[cc]{0,5}}

\textcolor{blue}{
\linethickness{0.3mm}
\qbezier(105,45)(105.38,45.46)(110.28,49.33)
\qbezier(110.28,49.33)(115.17,53.2)(120,55)
\qbezier(120,55)(129.9,57.99)(140,57.99)
\qbezier(140,57.99)(150.1,57.99)(160,55)
\qbezier(160,55)(164.83,53.2)(169.72,49.33)
\qbezier(169.72,49.33)(174.62,45.46)(175,45)
}

\textcolor{red}{
\linethickness{0.3mm}
\qbezier(105,50)(105.87,50.46)(115.97,54.35)
\qbezier(115.97,54.35)(126.07,58.23)(135,60)
\qbezier(135,60)(137.59,60.19)(140.19,59.54)
\qbezier(140.19,59.54)(142.8,58.88)(145,60)
\qbezier(145,60)(147.18,61.68)(147.69,64.73)
\qbezier(147.69,64.73)(148.21,67.78)(150,70)
\qbezier(150,70)(152.05,72.01)(154.61,73.19)
\qbezier(154.61,73.19)(157.17,74.36)(160,75)
\qbezier(160,75)(162.5,75.52)(164.96,75.3)
\qbezier(164.96,75.3)(167.42,75.07)(170,75)
}
\end{picture}

%% file: impTAstrong.tex
\ifx\JPicScale\undefined\def\JPicScale{0.85}\fi
\unitlength \JPicScale mm
\begin{picture}(185,100)(0,0)
\linethickness{0.3mm}
\put(15,80){\circle{10}}

\linethickness{0.3mm}
\put(35,90){\circle{10}}

\linethickness{0.3mm}
\put(55,90){\circle{10}}

\linethickness{0.3mm}
\put(75,80){\circle{10}}

\linethickness{0.3mm}
\put(35,70){\circle{10}}

\linethickness{0.3mm}
\put(55,70){\circle{10}}

\linethickness{0.3mm}
\put(175,80){\circle{10}}

\linethickness{0.3mm}
\put(155,90){\circle{10}}

\linethickness{0.3mm}
\put(135,90){\circle{10}}

\linethickness{0.3mm}
\put(135,70){\circle{10}}

\linethickness{0.3mm}
\put(155,70){\circle{10}}

\linethickness{0.3mm}
\put(115,80){\circle{10}}

\linethickness{0.3mm}
\put(115,25){\circle{10}}

\linethickness{0.3mm}
\put(135,35){\circle{10}}

\linethickness{0.3mm}
\put(155,35){\circle{10}}

\linethickness{0.3mm}
\put(175,25){\circle{10}}

\linethickness{0.3mm}
\put(155,15){\circle{10}}

\linethickness{0.3mm}
\put(135,15){\circle{10}}

\linethickness{0.3mm}
\put(15,25){\circle{10}}

\linethickness{0.3mm}
\put(35,35){\circle{10}}

\linethickness{0.3mm}
\put(55,35){\circle{10}}

\linethickness{0.3mm}
\put(75,25){\circle{10}}

\linethickness{0.3mm}
\put(55,15){\circle{10}}

\linethickness{0.3mm}
\put(35,15){\circle{10}}

\put(15,80){\makebox(0,0)[cc]{r}}

\put(115,80){\makebox(0,0)[cc]{r}}

\put(115,25){\makebox(0,0)[cc]{r}}

\put(15,25){\makebox(0,0)[cc]{r}}

\put(35,90){\makebox(0,0)[cc]{u}}

\put(135,90){\makebox(0,0)[cc]{u}}

\put(135,35){\makebox(0,0)[cc]{u}}

\put(35,35){\makebox(0,0)[cc]{u}}

\put(155,90){\makebox(0,0)[cc]{v}}

\put(155,35){\makebox(0,0)[cc]{v}}

\put(55,35){\makebox(0,0)[cc]{v}}

\put(55,90){\makebox(0,0)[cc]{v}}

\put(75,80){\makebox(0,0)[cc]{b}}

\put(175,80){\makebox(0,0)[cc]{b}}

\put(175,25){\makebox(0,0)[cc]{b}}

\put(75,25){\makebox(0,0)[cc]{b}}

\put(35,70){\makebox(0,0)[cc]{u'}}

\put(135,70){\makebox(0,0)[cc]{u'}}

\put(135,15){\makebox(0,0)[cc]{u'}}

\put(35,15){\makebox(0,0)[cc]{u'}}

\put(55,70){\makebox(0,0)[cc]{v'}}

\put(155,70){\makebox(0,0)[cc]{v'}}

\put(155,15){\makebox(0,0)[cc]{v'}}

\put(55,15){\makebox(0,0)[cc]{v'}}

\put(5,90){\makebox(0,0)[cc]{$\rho_0$}}

\put(105,90){\makebox(0,0)[cc]{$\rho_1$}}

\put(5,35){\makebox(0,0)[cc]{$\rho_2$}}

\put(105,35){\makebox(0,0)[cc]{$\rho_3$}}

\linethickness{0.3mm}
\multiput(120,80)(0.12,0.12){83}{\line(1,0){0.12}}
\put(120,80){\vector(-1,-1){0.12}}
\linethickness{0.3mm}
\multiput(120,80)(0.12,-0.12){83}{\line(1,0){0.12}}
\put(120,80){\vector(-1,1){0.12}}
\linethickness{0.3mm}
\multiput(160,90)(0.12,-0.12){83}{\line(1,0){0.12}}
\put(170,80){\vector(1,-1){0.12}}
\linethickness{0.3mm}
\multiput(160,70)(0.12,0.12){83}{\line(1,0){0.12}}
\put(170,80){\vector(1,1){0.12}}
\linethickness{0.3mm}
\multiput(20,25)(0.12,0.12){83}{\line(1,0){0.12}}
\put(20,25){\vector(-1,-1){0.12}}
\linethickness{0.3mm}
\multiput(20,25)(0.12,-0.12){83}{\line(1,0){0.12}}
\put(20,25){\vector(-1,1){0.12}}
\linethickness{0.3mm}
\put(40,35){\line(1,0){10}}
\put(40,35){\vector(-1,0){0.12}}
\linethickness{0.3mm}
\put(40,15){\line(1,0){10}}
\put(40,15){\vector(-1,0){0.12}}
\linethickness{0.3mm}
\multiput(60,35)(0.12,-0.12){83}{\line(1,0){0.12}}
\put(60,35){\vector(-1,1){0.12}}
\linethickness{0.3mm}
\multiput(120,25)(0.12,0.12){83}{\line(1,0){0.12}}
\put(120,25){\vector(-1,-1){0.12}}
\linethickness{0.3mm}
\multiput(120,25)(0.12,-0.12){83}{\line(1,0){0.12}}
\put(120,25){\vector(-1,1){0.12}}
\linethickness{0.3mm}
\put(140,35){\line(1,0){10}}
\put(140,35){\vector(-1,0){0.12}}
\linethickness{0.3mm}
\put(140,15){\line(1,0){10}}
\put(140,15){\vector(-1,0){0.12}}
\linethickness{0.3mm}
\multiput(20,80)(0.12,0.12){83}{\line(1,0){0.12}}
\linethickness{0.3mm}
\multiput(20,80)(0.12,-0.12){83}{\line(1,0){0.12}}
\linethickness{0.3mm}
\put(40,90){\line(1,0){10}}
\linethickness{0.3mm}
\put(40,70){\line(1,0){10}}
\linethickness{0.3mm}
\multiput(60,90)(0.12,-0.12){83}{\line(1,0){0.12}}
\linethickness{0.3mm}
\multiput(60,70)(0.12,0.12){83}{\line(1,0){0.12}}
\linethickness{0.3mm}
\put(140,90){\line(1,0){10}}
\linethickness{0.3mm}
\put(140,70){\line(1,0){10}}
\linethickness{0.3mm}
\multiput(60,15)(0.12,0.12){83}{\line(1,0){0.12}}
\linethickness{0.3mm}
\multiput(160,35)(0.12,-0.12){83}{\line(1,0){0.12}}
\linethickness{0.3mm}
\multiput(160,15)(0.12,0.12){83}{\line(1,0){0.12}}
\put(23,87.29){\makebox(0,0)[cc]{w}}

\put(22.29,71.71){\makebox(0,0)[cc]{w}}

\put(67.71,87.71){\makebox(0,0)[cc]{w}}

\put(67.43,72.57){\makebox(0,0)[cc]{w}}

\put(121.71,87.14){\makebox(0,0)[cc]{w}}

\put(168.29,87.71){\makebox(0,0)[cc]{w}}

\put(122,72.57){\makebox(0,0)[cc]{w}}

\put(168,71.43){\makebox(0,0)[cc]{w}}

\put(122.29,32.71){\makebox(0,0)[cc]{w}}

\put(167.43,33.57){\makebox(0,0)[cc]{w}}

\put(122.86,16.29){\makebox(0,0)[cc]{w}}

\put(167.14,15.86){\makebox(0,0)[cc]{w}}

\put(68.57,34.14){\makebox(0,0)[cc]{w}}

\put(68,16.14){\makebox(0,0)[cc]{w}}

\put(23,33.57){\makebox(0,0)[cc]{w}}

\put(22,16.43){\makebox(0,0)[cc]{w}}

\put(45,95){\makebox(0,0)[cc]{w'}}

\put(45,65){\makebox(0,0)[cc]{w'}}

\put(145,95){\makebox(0,0)[cc]{w'}}

\put(145,65){\makebox(0,0)[cc]{w'}}

\put(145,40){\makebox(0,0)[cc]{w'}}

\put(145,10){\makebox(0,0)[cc]{w'}}

\put(45,40){\makebox(0,0)[cc]{w'}}

\put(45,10){\makebox(0,0)[cc]{w'}}

\put(5,75){\makebox(0,0)[cc]{mr}}

\put(85,75){\makebox(0,0)[cc]{mr}}

\put(60,100){\makebox(0,0)[cc]{?}}

\put(60,60){\makebox(0,0)[cc]{?}}

\put(30,60){\makebox(0,0)[cc]{?}}

\put(30,100){\makebox(0,0)[cc]{?}}

\put(105,75){\makebox(0,0)[cc]{mr}}

\put(185,75){\makebox(0,0)[cc]{mr}}

\put(160,60){\makebox(0,0)[cc]{m}}

\put(160,100){\makebox(0,0)[cc]{m}}

\put(130,100){\makebox(0,0)[cc]{m}}

\put(130,60){\makebox(0,0)[cc]{m}}

\put(5,20){\makebox(0,0)[cc]{mr}}

\put(30,45){\makebox(0,0)[cc]{m}}

\put(30,5){\makebox(0,0)[cc]{m}}

\put(60,45){\makebox(0,0)[cc]{m'}}

\put(60,5){\makebox(0,0)[cc]{m'}}

\put(85,20){\makebox(0,0)[cc]{m''}}

\put(105,20){\makebox(0,0)[cc]{mr}}

\put(130,45){\makebox(0,0)[cc]{m}}

\put(130,5){\makebox(0,0)[cc]{m}}

\put(160,45){\makebox(0,0)[cc]{m'}}

\put(160,5){\makebox(0,0)[cc]{m'}}

\put(185,20){\makebox(0,0)[cc]{mr}}

\end{picture}

%% file: DuboisMasuzawaTixeuil.bbl
\begin{thebibliography}{10}

\bibitem{DD05c}
Ariel Daliot and Danny Dolev.
\newblock Self-stabilization of byzantine protocols.
\newblock In Ted Herman and S\'{e}bastien Tixeuil, editors, {\em
  Self-Stabilizing Systems}, volume 3764 of {\em Lecture Notes in Computer
  Science}, pages 48--67. Springer, 2005.

\bibitem{D74j}
Edsger~W. Dijkstra.
\newblock Self-stabilizing systems in spite of distributed control.
\newblock {\em Commun. ACM}, 17(11):643--644, 1974.

\bibitem{D00b}
Shlomi. Dolev.
\newblock {\em Self-stabilization}.
\newblock MIT Press, March 2000.

\bibitem{DW04j}
Shlomi Dolev and Jennifer~L. Welch.
\newblock Self-stabilizing clock synchronization in the presence of byzantine
  faults.
\newblock {\em J. ACM}, 51(5):780--799, 2004.

\bibitem{DMT10ca}
Swan Dubois, Toshimitsu Masuzawa, and S\'{e}bastien Tixeuil.
\newblock The impact of topology on byzantine containment in stabilization.
\newblock In {\em Proceedings of DISC 2010}, {L}ecture {N}otes in {C}omputer
  {S}cience, Boston, Massachusetts, USA, September 2010. {S}pringer {B}erlin /
  {H}eidelberg.

\bibitem{DMT10cd}
Swan Dubois, Toshimitsu Masuzawa, and S\'{e}bastien Tixeuil.
\newblock On byzantine containment properties of the min+1 protocol.
\newblock In {\em Proceedings of SSS 2010}, {L}ecture {N}otes in {C}omputer
  {S}cience, New York, NY, USA, September 2010. {S}pringer {B}erlin /
  {H}eidelberg.

\bibitem{DMT11j}
Swan Dubois, Toshimitsu Masuzawa, and S\'{e}bastien Tixeuil.
\newblock Bounding the impact of unbounded attacks in stabilization.
\newblock {\em IEEE Transactions on Parallel and Distributed Systems (TPDS)},
  2011.

\bibitem{GS99c}
Mohamed~G. Gouda and Marco Schneider.
\newblock Stabilization of maximal metric trees.
\newblock In Anish Arora, editor, {\em WSS}, pages 10--17. IEEE Computer
  Society, 1999.

\bibitem{GS03j}
Mohamed~G. Gouda and Marco Schneider.
\newblock Maximizable routing metrics.
\newblock {\em IEEE/ACM Trans. Netw.}, 11(4):663--675, 2003.

\bibitem{HC92j}
Shing-Tsaan Huang and Nian-Shing Chen.
\newblock A self-stabilizing algorithm for constructing breadth-first trees.
\newblock {\em Inf. Process. Lett.}, 41(2):109--117, 1992.

\bibitem{LSP82j}
Leslie Lamport, Robert~E. Shostak, and Marshall~C. Pease.
\newblock The byzantine generals problem.
\newblock {\em ACM Trans. Program. Lang. Syst.}, 4(3):382--401, 1982.

\bibitem{MT06cb}
Toshimitsu Masuzawa and S\'{e}bastien Tixeuil.
\newblock Bounding the impact of unbounded attacks in stabilization.
\newblock In Ajoy~Kumar Datta and Maria Gradinariu, editors, {\em SSS}, volume
  4280 of {\em Lecture Notes in Computer Science}, pages 440--453. Springer,
  2006.

\bibitem{MT07j}
Toshimitsu Masuzawa and S\'{e}bastien Tixeuil.
\newblock Stabilizing link-coloration of arbitrary networks with unbounded
  byzantine faults.
\newblock {\em International Journal of Principles and Applications of
  Information Science and Technology (PAIST)}, 1(1):1--13, December 2007.

\bibitem{NA02c}
Mikhail Nesterenko and Anish Arora.
\newblock Tolerance to unbounded byzantine faults.
\newblock In {\em 21st Symposium on Reliable Distributed Systems (SRDS 2002)},
  page~22. IEEE Computer Society, 2002.

\bibitem{T09bc}
S\'{e}bastien Tixeuil.
\newblock {\em Algorithms and Theory of Computation Handbook, Second Edition},
  chapter Self-stabilizing Algorithms, pages 26.1--26.45.
\newblock Chapman \& Hall/CRC Applied Algorithms and Data Structures. CRC
  Press, Taylor \& Francis Group, November 2009.

\end{thebibliography}
